\documentclass{sig-alternate}
\usepackage[bookmarks=false]{hyperref}
\usepackage{wrapfig}
\usepackage{amsmath,amsfonts,amssymb}
\usepackage{euscript}
\usepackage{epsfig}
\usepackage{xspace}
\usepackage{relsize}
\usepackage{colortbl}
\usepackage[update,prepend]{epstopdf}
\usepackage{subfigure,color}
\usepackage{theorem}
\usepackage{graphicx}
\makeatletter
\newif\if@restonecol
\makeatother

\usepackage[section]{algorithm}
\usepackage[noend]{algorithmic}
\usepackage{cite}

\newtheorem{lemma}{Lemma}

\newtheorem{theorem}{Theorem}

\usepackage{cite}
\usepackage[table]{xcolor}

{\theorembodyfont{\rmfamily}
\newtheorem{definition}{Definition}}
{\theorembodyfont{\rmfamily}
}
\newenvironment{pkl}{%
\begin{itemize}%
\setlength\itemsep{-0.5\parskip}%
\setlength\parsep{0in}%
}{%
\end{itemize}}

\newcommand{\etal}{\emph{et.al.}}

\newcommand{\Paragraph}[1]{~\vspace*{-0.9\baselineskip}\\{\bf #1}}

\newlength{\figsize} 

\def\R{\mathbb{R}}

\newcommand{\eps}{\varepsilon}

\renewcommand{\b}[1]{\ensuremath{\mathbb{#1}}}

\newcommand{\denselist}{\itemsep -2pt\parsep=-1pt\partopsep -2pt}

\newcommand{\omt}[1]{}

\newcommand{\s}[1]{{\textsf{#1}}}

\renewcommand{\b}[1]{\ensuremath{\mathbb{#1}}}

\newcommand{\E}{\textbf{\s{E}}}
\newcommand{\Var}{\textbf{\s{Var}}}
\renewcommand{\Pr}{\textbf{\s{Pr}}}

\newcommand{\svd}{\s{svd}}
\newcommand{\diag}{\s{diag}}

\newcommand{\rank}{\s{rank}}

\newcommand{\unif}{\s{Unif}}
\newcommand{\pluseq}{\mathrel{+}=}

\newcommand{\piw}{\s{P1}}
\newcommand{\piiw}{\s{P2}}
\newcommand{\piiiw}{\s{P3}}
\newcommand{\piiiiw}{\s{P4}}
\newcommand{\pim}{\s{P1}}
\newcommand{\piim}{\s{P2}}
\newcommand{\piiim}{\s{P3}}

\begin{document}

\title{Continuous Matrix Approximation on Distributed Data\thanks{Thanks to support from NSF grants CCF-1115677, IIS-1251019, IIS-1200792, and IIS-1251019.}}

\numberofauthors{3} 

\author{
\alignauthor
Mina Ghashami\\
       \affaddr{School of Computing}\\
       \affaddr{University of Utah}\\
       \email{ghashami@cs.uah.edu}
\alignauthor
Jeff M. Phillips\\
       \affaddr{School of Computing}\\
       \affaddr{University of Utah}\\
       \email{jeffp@cs.uah.edu}
\alignauthor Feifei Li \\
       \affaddr{School of Computing}\\
       \affaddr{University of Utah}\\
       \email{lifeifei@cs.uah.edu}
}

\maketitle

Tracking and approximating data matrices in streaming fashion is a
fundamental challenge. The problem requires more care and attention
when data comes from multiple distributed sites, each receiving a
stream of data.
This paper considers the problem of ``tracking approximations to a
matrix'' in the distributed streaming model. In this model, there are $m$
distributed sites each observing a distinct stream of data (where each
element is a row of a distributed matrix) and has a communication
channel with a coordinator, and the goal is to track an
$\eps$-approximation to the norm of the matrix along any direction. To that end,
we present novel algorithms to address the matrix approximation
problem. Our algorithms maintain a smaller matrix $B$, as an
approximation to a distributed streaming matrix $A$, such that for any
unit vector $x$: $| \|A x\|^2 - \|B x\|^2 | \leq \eps \|A\|_F^2$. Our
algorithms work in streaming fashion and incur small communication,
which is critical for distributed computation. Our best method is
deterministic and uses only $O((m/\eps) \log(\beta N))$ communication,
where $N$ is the size of stream (at the time of the query) and $\beta$
is an upper-bound on the squared norm of any row of the matrix. 
In addition to proving all algorithmic properties
theoretically, extensive experiments with real large datasets
demonstrate the efficiency of these protocols.

\section{Introduction}
\label{sec:intro}
Large data matrices are found in numerous domains, such as scientific
computing, multimedia applications, networking systems, server and
user logs, and many others
\cite{Lib12,DBLP:conf/vldb/PapadimitriouSF05,DBLP:conf/sigcomm/LakhinaCD04,DBLP:conf/sigmetrics/LakhinaPCDKT04}. 
Since such data is huge in size and often generated continuously, it is important to process them in streaming fashion and maintain an approximating summary.
Due to its importance, the matrix approximation problem has received
careful investigations in the literature; the latest significant effort is represented by Liberty~\cite{Lib12} on a centralized stream.

In recent years, {\em distributed streaming model}
\cite{cormode2011algorithms} has become popular, in which there are multiple distributed
sites, each observing a disjoint stream of data and together attempt to monitor a function at a single coordinator site $C$. Due
to its wide applications in practice \cite{cormode2011algorithms}, a
flurry of work has been done under this setting.
This model is more general than the \textit{streaming model}
\cite{babcock2002models} that maintains a function at a single site
with small space. It is also different from \textit{communication
  model} \cite{yao1979some} in which data is (already) stored at
multiple sites and the goal is to do a \textit{one-time} computation
of a target function.  The key resources to optimize in distributed
streaming model is not just the space needed at the coordinator or
each site, but the communication between the sites and coordinator.

Despite prior works on distributed streaming model, and distributed matrix computations (e.g., the MadLINQ library \cite{DBLP:conf/eurosys/QianCKCYMZ12}), 
little is known on 
continuously tracking an matrix approximation in the distributed
streaming model. 
This paper considers the important
problem of {\em ``tracking approximations to matrices''} in the
distributed streaming model \cite{cormode2011algorithms}.

\Paragraph{Motivation.} Our problem is motivated by many applications
in distributed databases, wireless sensor networks, cloud computing,
etc \cite{muthukrishnan2005data} where data sources are distributed
over a network and collecting all data together at a central location
is not a viable option. In many such environments queries must be
answered continuously, based on the total data that has arrived so
far.

For example, in large scale image analysis, each row in the matrix
corresponds to one image and contains either pixel values or other
derived feature values (e.g, 128-dimensional SIFT features). A search
engine company has image data continuously arriving at many data
centers, or even within a single data center at many nodes in a massive
cluster. This forms a distributed matrix and it is critical to obtain
excellent, real-time approximation of the distributed streaming image
matrix with little communication overhead.

Yet another example is for large-scale distributed web crawling or
server access log monitoring/mining, where data in the bag-of-words
model is a matrix whose columns correspond to words or tags/labels
(for textual analysis, e.g. LSI, and/or for learning and classification purpose)
and rows correspond to documents or log records (which arrive
continuously at distributed nodes).

Since data is continuously changing in these applications, query
results can also change with time.  So the challenge is to
minimize the communication between sites and the coordinator
while maintaining accuracy of results {\em at all time}.

In our case, 
each site may generate a record in a given time instance, which is a
row of a matrix. The goal is to {\em approximate the matrix that is
  the union of all the rows from all sites until current time instance
  $t_{\textrm{now}}$ continuously} at the coordinator $C$. This
problem can be easily found in distributed network monitoring
applications \cite{DBLP:conf/imc/PapagiannakiTL04}, distributed data
mining, cloud computing \cite{DBLP:conf/eurosys/QianCKCYMZ12}, stream
mining \cite{DBLP:conf/kdd/HultenSD01}, and log analysis from multiple
data centers
\cite{DBLP:journals/tocs/CorbettDEFFFGGHHHKKLLMMNQRRSSTWW13}.

\Paragraph{Distributed streaming matrix approximation.} Formally,
assume $A = (a_1, \ldots, a_n, \ldots)$ is an unbounded stream of
items.  At the current time $t_{\textrm{now}}$, let $n$ denote the
number of items the system has seen so far; that is at the current
time the dataset is $A = (a_1, \ldots, a_n)$.  And although we do not
place a bound on the number of items, we let $N$ denote the total size
of the stream {\em at the time when a query $q$ is
  performed}.  This allows us to discuss results in terms of $n$ at a
given point, and in terms of $N$ for the entire run of the algorithm
until the time of a query $q$.

At each time step we assume the item $a_n$ appears at exactly one of
$m$ sites $S_1, \ldots, S_m$.  The goal is to approximately maintain
or monitor some function $f(A)$ of $A$ at the coordinator node $C$.
Each site can only communicate with the coordinator (which in practice
may actually be one of the sites).  This model is only off by a factor
$2$ in communication with a model where all pairs of sites can
communicate, by routing through the coordinator.

Our goal is to maintain a value $\hat f(A)$ which is off by some
$\eps$-factor from the true value $f(A)$.  The function and manner of
approximation is specified for a variety of
problems~\cite{yi2013optimal, manjhi2005finding,
  keralapura2006communication, babcock2003distributed,
  cormode2005holistic, cormode2005sketching} included total count,
heavy hitters, quantiles, and introduced in this paper, matrix
approximation:

\vspace{-4mm}
\begin{definition} [{\small Tracking distributed streaming matrix}]:\\*
  Suppose each $a_n$ is a record with $d$ attributes, a row from a
  matrix in an application. Thus, at time $t_{\textrm{now}}$, $A =
  (a_1, \ldots, a_n)$ forms a $n\times d$ {\em distributed streaming
    matrix}. At any time instance, $C$ needs to approximately maintain 
  the norm of matrix $A$ along any arbitrary direction. The goal is to
  continuously track a small approximation of matrix $A$. Formally,
  for any time instance $t_{\textrm{now}}$ (i.e., for any $n$), $C$
  needs to maintain a smaller matrix $B\in \mathbb{R}^{\ell\times d}$
  as an approximation to the distributed streaming matrix $A\in
  \mathbb{R}^{n\times d}$ such that $\ell\ll n$ and for any unit
  vector $x$: $| \|Ax\|^2 - \|Bx\|^2 | \leq \eps \|A\|_F^2$.

  As required in the distributed streaming model
  \cite{cormode2011algorithms}, each site $S_i$ must process its
  incoming elements in streaming fashion. The objective is to minimize
  the total communication between $C$ and all sites $S_1, \ldots,
  S_m$. \qed
\end{definition}

\vspace{-3mm} 
\Paragraph{Additional notations.} In the above definition, the
Frobenius norm of a matrix $A$ is $\| A\|_F = \sqrt{\sum_{i=1}^n \|
  a_i\|^2}$ where $\|a_i\|$ is standard Euclidean norm of row $a_i$;
the Frobenius norm is a widely used matrix norm. We also
let $A_k$ be the {\em best rank $k$ approximation} of matrix $A$,
specifically $A_k = {\arg \min}_{ X: \rank(X) \leq k} \| A - X\|_F$.

\Paragraph{Our contributions.}
This work makes important contributions in solving the open problem of
tracking distributed streaming matrix. Instead of exploring heuristic
methods that offer no guarantees on approximation quality, we focus on
principled approaches that are built on sound theoretical
foundations. Moreover, all methods are simple and efficient to
implement as well as effective in practice. Specifically:
\vspace{-2mm}
\begin{pkl}
\item We establish and further the important connection between
  tracking a matrix approximation and maintaining $\eps$-approximate
  {\em weighted} heavy-hitters in the distributed streaming model in
  Section \ref{sec:weightedModel}.  Initial insights along this direction
  were established in very recent work~\cite{Lib12,GP14}; but the
  extent of this connection is still limited and not fully understood.
  This is demonstrated, for instance, by us showing how three
  approaches for weighted heavy hitters can be adapted to matrix
  sketching, but a fourth cannot.

\item We introduce four new methods for tracking the
  $\eps$-approximate {\em weighted} heavy-hitters in a distributed
  stream in Section \ref{sec:weightedModel}, and analyze their
  behaviors with rigorous theoretical analysis.

\item We design three novel algorithms for tracking a good
  approximation of a distributed streaming matrix in Section
  \ref{sec:matrix}; these leverage the new insights connecting this
  problem to solutions in Section \ref{sec:weightedModel} for
  distributed weighted heavy hitters tracking.

\item We present thorough experimental evaluations of the proposed
  methods in Section \ref{sec:exp} on a number of large real data
  sets. Experimental results verify the effectiveness of our methods
  in practice.
\end{pkl} \vspace{-2mm}

In addition, we provide a thorough review of related works and
relevant background material in Section \ref{sec:related}.  The paper
is concluded in Section \ref{sec:conclude}.

\section{Related work}
\label{sec:related}
There are two main classes of prior work that are relevant to our
study: approximating a streaming matrix in a centralized stream, and
tracking heavy hitters in either a centralized stream or (the union
of) distributed steams.

\Paragraph{Matrix approximation in a centralized stream.} Every
incoming item in a centralized stream represents a new row of data in
a streaming matrix. The goal is to continuously maintain a low rank
matrix approximation. It is a special instance of our problem for
$m=1$, i.e., there is only a single site. Several results exist in the
literature, including streaming PCA (principal component analysis)
\cite{mitliagkas2013memory}, streaming SVD (singular value
decomposition) \cite{Strumpen03astream,Brand200620}, 
and matrix sketching 
\cite{clarkson2009numerical,Lib12,GP14}. The matrix sketching
technique \cite{Lib12} only recently appeared and is the
start-of-the-art for low-rank matrix approximation in a single
stream. Liberty \cite{Lib12} adapts a well-known streaming algorithm for
approximating item frequencies, the MG algorithm \cite{mg-fre-82}, to
sketching a streaming matrix. The method, Frequent
Directions (\textsf{FD}), receives $n$ rows of a matrix $A\in
\mathbb{R}^{n\times d}$ one after another, in a centralized streaming
fashion. It maintains a sketch $B\in \mathbb{R}^{\ell\times d}$ with
only $\ell\ll n$ rows, but guarantees that $A^TA\approx B^TB$. More
precisely, it guarantees that $\forall x\in \b{R}^{d}$, $\| x\|=1$, $0\le
\|Ax\|^2 - \|Bx\|^2 \le 2\| A\|^2_F/\ell$. \textsf{FD} uses $O(d\ell)$
space, and each item updates the sketch in amortized $O(d\ell)$ time;
two such sketches can also be merged in $O(d \ell^2)$ time.

A bound on $\|A x\|$ for any unit vector $x$ preserves norm (or
length) of a matrix in direction $x$.  For instance, when one
performs PCA on $A$, it returns the set of the top $k$ orthogonal
directions, measured in this length.  These are the linear
combinations of attributes (here we have $d$ such attributes) which
best capture the variation within the data in $A$.  Thus by increasing
$\ell$, this bound allows one to approximately retain all important
linear combinations of attributes.

An extension of \textsf{FD} to derive streaming sketch results with
bounds on relative errors, i.e., to ensure that $\|A-A_k\|_F^2\le
\|A\|_F^2-\|B_k\|_F^2\le (1+\eps)\|A-A_k\|_F^2$, appeared in
\cite{GP14}. It also gives that $\|A-\pi_{B_k}(A)\|_F^2\le
(1+\eps)\|A-A_k\|^2_F$ where $B_k$ is the top $k$ rows of $B$ and
$\pi_{B_k}(A)$ is the projection of $A$ onto the row-space of $B_k$.
This latter bound is interesting because, as we will see, it indicates
that when most of the variation is captured in the first $k$ principal
components, then we can almost recover the entire matrix exactly.

But none of these results can be applied in distributed streaming
model without incurring high communication cost. Even though the
\textsf{FD} sketches are mergeable, since the coordinator $C$ needs to
maintain an approximation matrix {\em continuously} for all time
instances. One has to either send $m$ sketches to $C$, one from each
site, at every time instance and ask $C$ to merge them to a single
sketch, or one can send a streaming element to $C$ whenever it is
received at a site and ask $C$ to maintain a single sketch using
\textsf{FD} at $C$. Either approach will lead to $\Omega(N)$ communication.

Nevertheless, the design of \textsf{FD} has inspired us to explore the
connection between distributed matrix approximation and approximation
of distributed heavy hitters. 

\Paragraph{Heavy hitters in distributed streams.}  Tracking heavy
hitters in distributed streaming model is a fundamental problem
\cite{yi2013optimal, manjhi2005finding, keralapura2006communication,
  babcock2003distributed}.  Here we assume each $a_n \in A$ is an
element of a bounded universe $[u] = \{1, \ldots, u\}$.  If we denote
the frequency of an element $e \in [u]$ in the stream $A$ as $f_e(A)$,
the $\phi$-heavy hitters of $A$ (at time instance $t_{\textrm{now}}$)
would be those items $e$ with $f_e(A) \geq \phi n$ for some parameter
$\phi \in [0,1]$.  We denote this set as $H_\phi(A)$.  Since computing
exact $\phi$-heavy hitters incurs high cost and is often unnecessary,
we allow an $\eps$-approximation, then the returned set of heavy
hitters must include $H_\phi(A)$, may or may not include items $e$
such that $(\phi - \eps) n \leq f_e(A) < \phi n$ and must not include
items $e$ with $f_e(A) < (\phi - \eps) n$.

Babcock and Olston~\cite{babcock2003distributed} designed some
deterministic heuristics called as \emph{top-$k$ monitoring} to
compute top-$k$ frequent items.
Fuller and Kantardzid modified their technique and proposed
\emph{FIDS} \cite{fuller2007fids}, a heuristic
method, 
to track the heavy hitters while reducing communication cost and
improving overall quality of results. Manjhi \etal
\cite{manjhi2005finding} studied $\phi$-heavy hitter tracking in a
hierarchical communication model.

Cormode and Garofalakis~\cite{cormode2005sketching} proposed another
method by maintaining a summary of the input stream and a prediction
sketch at each site. If the summary varies from the prediction sketch
by more than a user defined tolerance amount, the summary and
(possibly) a new prediction sketch is sent to a coordinator. The
coordinator can use the information gathered from each site to
continuously report frequent items. Sketches maintained by each site
in this method require $O((1/\eps^2) \log(1/\delta))$ space and
$O(\log(1/\delta))$ time per update, where $\delta$ is a probabilistic
confidence.

Yi and Zhang~\cite{yi2013optimal} provided a deterministic algorithm
with communication cost $O((m/\eps) \log N)$ and $O(1/\eps)$ space at
each site to continuously track $\phi$-heavy hitters and the
$\phi$-quantiles. In their method, every site and the coordinator have as
many counters as the type of items plus one more counter for the total
items.
Every site keeps track of number of items it receives in each round,
once this number reaches roughly $\eps/m$ times of the total counter
at the coordinator, the site sends the counter to the coordinator. After the
coordinator receives $m$ such messages, it updates its counters and
broadcasts them to all sites. Sites reset their counter values and
continue to next round. To lower space usage at sites, they suggested
using space-saving sketch\cite{metwally2006integrated}.  The authors
also gave matching lower bounds on the communication costs for both
problems, showing their algorithms are optimal in the deterministic
setting.

Later, Huang \etal ~\cite{huang2012randomized} proposed a randomized
algorithm that uses $O(1/(\eps \sqrt{m}))$ space at each site and
$O((\sqrt{m}/\eps) \log N)$ total communication and tracks heavy
hitters in a distributed stream.  For each item $a$ in the stream a
site chooses to send a message with a probability $p = \sqrt{m}/(\eps
\hat n)$ where $\hat n$ is a $2$-approximation of the total count.  It
then sends $f_e(A_j)$ the total count of messages at site $j$ where
$a=e$, to the coordinator.  Again an approximation heavy-hitter count
$\hat f_e(A_j)$ can be used at each site to reduce space.

The $\eps$-heavy hitters can be maintained from a random sampling of
elements of size $s = O(1/\eps^2)$.  This allows one to use the well
studied technique of maintaining a random sample of size $s$ from a
distributed stream~\cite{CMYZ12,TW11}, which can be done with roughly
$O((m + s) \log (N/s))$ communication.

\Paragraph{Other related work.} Lastly, our work falls into the
general problem of tracking a function in distributed streaming
model. Many existing works have studied this general problem for
various specific functions, and we have reviewed the most related ones
on heavy hitters. A detailed survey of results on other functions
(that are much less relevant to our study) is beyond the scope of this
work, and we refer interested readers to
\cite{cormode2011algorithms,Cor13} and references therein.

Our study is also related to various matrix computations over
distributed matrices, for example, the MadLINQ library
\cite{DBLP:conf/eurosys/QianCKCYMZ12}. However, these results focus on
one-time computation over non-streaming data, i.e., in the {\em
  communication model} \cite{yao1979some} as reviewed in Section
\ref{sec:intro}. We refer interested readers to
\cite{DBLP:conf/eurosys/QianCKCYMZ12} and references therein for
details.

There are also many studies on low-rank approximations of matrices in
centralized, non-streaming setting, e.g.,
\cite{liberty2007randomized,frieze2004fast,achlioptas2001fast,drineas2003pass}
and others, but these methods are not applicable for a distributed
streaming setting.

\vspace{-2mm}
\section{Insights, Weights and Rounds} 
\label{sec:weights}
A key contribution of this paper is strengthening the connection
between maintaining approximate weighted frequency counts and approximately maintaining matrices.

To review, the Misra-Gries (MG) algorithm~\cite{mg-fre-82} is a deterministic, associative sketch to approximate frequency counts, in contrast to, say the popular count-min sketch~\cite{cormode2005improved} which is randomized and hash-based.  
MG maintains an associative array of size $\ell$ whose keys are elements of $e \in [u]$, and values are estimated frequency $\hat f_e$ such that $0 \leq f_e - \hat f_e \leq n/\ell$. Upon processing element $e \in A$, three cases can occur. If $e$ matches a label, it increments the associated counter. If not, and there is an empty counter, it sets the label of the counter to $e$ and sets its counter to $1$.  Otherwise, if no empty counters, then it decrements (shrinks) all counters by $1$. 

Liberty~\cite{Lib12} made this connection between frequency estimates and matrices in a centralized stream through the singular value decomposition (\svd).  
Our approaches avoid the \svd\ or use it in a different way.
The \svd\ of an $n \times d$ matrix $A$ returns (among other things) a
set of $d$ singular values $\{\sigma_1 \geq \sigma_2 \geq \ldots \geq
\sigma_d\}$ and a corresponding set of orthogonal right singular
vectors $\{v_1, \ldots, v_d\}$.  It holds that $\|A\|_F^2 =
\sum_{i=1}^d \sigma_i^2$ and for any $x$ that $\|A x\|^2 =
\sum_{i=1}^d \sigma_i^2 \langle v_i, x \rangle^2$.  The work of
Liberty~\cite{Lib12} shows that one can run a version of the
Misra-Gries~\cite{mg-fre-82} sketch for frequency counts using the
singular vectors as elements and squared singular values as the corresponding
total weights, recomputing the \svd\ when performing the shrinkage
step.

To see why this works, let us consider the restricted case where every
row of the matrix $A$ is an indicator vector along the standard
orthonormal basis.  That is, each row $a_i\in\{e_1, \ldots, e_d\}$,
where $e_j = (0, \ldots, 0, 1, 0, \ldots, 0)$ is the $j$th standard
basis vector. Such a matrix $A$ can encode a stream $S$ of items from
a domain $[d]$. If the $i$th element in the stream $S$ is item
$j\in\{1, \ldots, d\}$, then the $i$th row of the matrix $A$ is set to
$a_i=e_j$. At time $t_{\textrm{now}}$, the frequency $f_j$ from $S$
can be expressed as $f_j=\|A e_j\|^2$, since $\|A x\|^2 = \sum_{i=1}^n
\langle a_i, x \rangle^2$ and the dot product is only non-zero in this
matrix for rows which are along $e_j$.  A good approximate matrix $B$
would be one such that $g_j=\|B e_j\|^2$ is a good approximation of
$f_j$. Given $\|A\|_F^2=n$ (since each row of $A$ is a standard basis
vector), we derive that $|f_j-g_j|\le \eps n$ is equivalent to 
$| \|A e_j\|^2-\|B e_j\|^2 | \le \eps \|A\|_F^2$.

However, general matrices deviate from this simplified example in having non-orthonormal rows.
Liberty's \textsf{FD} algorithm~\cite{Lib12} demonstrates how taking \svd\ of a general matrix gets around this deviation and it achieves the same
bound $| \|A x\|^2-\|B x\|^2 | \le \eps \|A\|_F^2$ for any unit vector $x$.  

Given the above connection, in this paper, first we propose four novel
methods for tracking {\em weighted heavy hitters} in a distributed
stream of items ({\em note that tracking {\em weighted} heavy hitters in
the distributed streaming model has not been studied before}). Then, we
try to extend them to track matrix approximations where
elements of the stream are $d$-dimensional vectors.  Three extensions
are successful, including one based directly on Liberty's algorithm,
and two others establishing new connections between weighted frequency
estimation and matrix approximation.  We also show why the fourth
approach cannot be extended, illustrating that the newly established
connections are not obvious.

\Paragraph{Upper bound on weights.}
As mentioned above, there are two main challenges in extending ideas
from frequent items estimation to matrix approximation.  While, the
second requires delving into matrix decomposition properties, the
first one is in dealing with weighted elements.  We discuss here some
high-level issues and assumptions we make towards this goal.

In the weighted heavy hitters problem (similarly in matrix
approximation problem) each item $a_i$ in the stream has a weight
$w_i$ (for matrices this weight will be implicit as $\|a_i\|^2$).  Let
$W = \sum_{i=1}^n w_i$ be the total weight of the problem.  However,
allowing arbitrary weights can cause problems as demonstrated in the
following example.

Suppose we want to maintain a $2$-approximation of the total weight
(i.e. a value $\hat W$ such that $\hat W \leq W \leq 2 \hat W$).  If
the weight of each item doubles (i.e. $w_i = 2^i$ for tuple $(a_i,
w_i) \in A$), every weight needs to be sent to the coordinator.  This
follows since $W$ more than doubles with every item, so $\hat W$
cannot be valid for more than one step.  The same issue arises in
tracking approximate heavy hitters and matrices.

To make these problems well-posed, often
researchers~\cite{hung2008finding} assume weights vary in a finite
range, and are then able to bound communication cost.  To this end we
assume all $w_i \in [1,\beta]$ for some constant $\beta \geq 1$.

One option for dealing with weights is to just pretend every item with
element $e$ and weight $w_i$ is actually a set of $\lceil w_i \rceil$
distinct items of element $e$ and weight $1$ (the last one needs to be
handled carefully if $w_i$ is not an integer).  But this can increase
the total communication and/or runtime of the algorithm by a factor
$\beta$, and is not desirable.

Our methods take great care to only increase the communication by a
$\log (\beta N)/\log N$ factor compared to similar unweighted
variants.  In unweighted version, each protocol proceeds in $O(\log
N)$ rounds (sometimes $O(\frac{1}{\eps}\log N)$ rounds); a new round
starts roughly when the total count $n$ doubles.  In our settings, the
rounds will be based on the total weight $W$, and will change roughly
when the total weight $W$ doubles.  Since the final weight $W \leq
\beta N$, this will causes an increase to $O(\log W) = O(\log (\beta
N))$ rounds.  The actual analysis requires much more subtlety and care
than described here, as we will show in this
paper. 
Next we first discuss the weighted heavy-hitters problem, as the
matrix tracking problem will directly build on these techniques, and
some analysis will also carry over directly.

\section{Weighted Heavy Hitters in A Distributed
  Stream} 
\label{sec:weightedModel}
The input is a {\em distributed} weighted data stream $A$, which is a
sequence of tuples $(a_1, w_1), (a_2, w_2),\ldots , (a_n, w_n),
\ldots$ where $a_n$ is an element label and $w_n$ is the weight. For
any element $e \in [u]$, define $A_e = \{(a_i, w_i) \mid a_i = e\}$
and let $W_e = \sum_{(a_i, w_i) \in A_e} w_i$.
For notational convenience, we sometimes refer to a tuple $(a_i, w_i)
\in A$ by just its element $a_i$.

There are numerous important motivating scenarios for this
extension. For example, instead of just monitoring counts of objects,
we can measure a total size associated with an object, such as total
number of bytes sent to an IP address, as opposed to just a count of
packets. 
We next describe how to extend four protocols for heavy hitters
to the weighted setting.  
These are extensions of the unweighted protocols described in Section \ref{sec:related}.

\Paragraph{Estimating total weight.} An important task is to
approximate the current total weight $W = \sum_{i=1}^n w_i$ for all
items across all sites.  This is a special case of the heavy hitters
problem where all items are treated as being the same element. So if
we can show a result to estimate the weight of any single element
using a protocol within $\eps W$, then we can get a global estimate
$\hat W$ such that $|W - \hat W| \leq \eps W$.  All our subsequent
protocols can run a separate process in parallel to return this
estimate if they do not do so already.

Recall that the heavy hitter problem typically calls to return all
elements $e \in [u]$ if $f_e(A)/W \geq \phi$, and never if $f_e(A)/W <
\phi-\eps$.  For each protocol we study, the main goal is to ensure
that an estimate $\hat W_e$ satisfies $|f_e(A) - \hat W_e| \leq \eps
W$.  We show this, along with the $\hat W$ bound above, adjusting constants, is sufficient
to estimate weighted heavy hitters.  We return $e$ as a
$\phi$-weighted heavy hitter if $\hat W_e / \hat W > \phi - \eps/2$.

\vspace{-2mm}
\begin{lemma}
  If $|f_e(A) - \hat W_e| \leq (\eps/6) W$ and $|W - \hat W| \leq
  (\eps/5) W$, we return $e$ if and only if it is a valid weighted
  heavy hitter.
\end{lemma}
\begin{proof} 
  We need $|\frac{\hat W_e}{\hat W} - \frac{f_e(A)}{W} | \leq \eps/2$.
  We show the upper bound, the lower bound argument is symmetric.
\begin{align*}
\frac{\hat W_e}{\hat W} 
&\leq 
\frac{f_e(A)}{\hat W} + \frac{\eps}{6} \frac{W}{\hat W}
\leq
\frac{f_e(A)}{W} \frac{1}{1-\eps/5} + \frac{\eps}{5} \frac{1+\eps/5}{1-\eps/5}
\\ &\leq
\frac{f_e(A)}{W} + \frac{\eps}{4} + \frac{\eps}{4}
=
\frac{f_e(A)}{W} + \frac{\eps}{2}. \qed
\end{align*}
\end{proof}

Given this result, we can focus just on approximating the frequency $f_e(A)$ of all items.

\subsection{Weighted Heavy Hitters, Protocol 1}
We start with an intuitive approach to the distributed streaming
problem: run a streaming algorithm (for frequency estimation) on each
site, and occasionally send the full summary on each site to the
coordinator.  We next formalize this protocol (P1).

On each site we run the Misha-Gries summary~\cite{mg-fre-82} for
frequency estimation, modified to handle weights, with $2/\eps =
1/\eps'$ counters.  We also keep track of the total weight $W_i$ of
all data seen on that site $i$ since the last communication with the
coordinator.  When $W_i$ reaches a threshold $\tau$, site $i$ sends
all of its summaries (of size only $O(m/\eps)$) to the coordinator.
We set $\tau = (\eps/2m) \hat W$, where $\hat W$ is an estimate of the
total weight across all sites, provided by the coordinator.  At this
point the site resets its content to empty.  This is summarized in
Algorithm \ref{alg:P1w-site}.

The coordinator can merge results from each site into a single summary
without increasing its error bound, due to the mergeability of such
summaries~\cite{ACHPWY12}.  It broadcasts the updated total weight
estimate $\hat W$ when it increases sufficiently since the last
broadcast.  See details in Algorithm \ref{alg:P1w-coord}.

\begin{algorithm}
\caption{\label{alg:P1w-site} P1: Tracking heavy-hitters (at site $S_i$)}
\begin{algorithmic}
  \FOR{$(a_n, w_n)$ in round $j$}
    \STATE Update $G_i \leftarrow \textsf{MG}_{\eps'}(G_i,(a_n, w_n))$.  
    \STATE Update total weight on site $W_i \pluseq w_n$.  
    \IF {($W_i \geq \tau = (\eps/2m) \hat W$)}
      \STATE Send $(G_i, W_i)$ to coordinator; make $G_i, W_i$ empty.  
    \ENDIF
  \ENDFOR
\end{algorithmic}
\end{algorithm}

\begin{algorithm}
\caption{\label{alg:P1w-coord} P1: Tracking heavy-hitters (at $C$)}
\begin{algorithmic}
  \STATE On input $(G_i, W_i)$:
  \STATE Update sketch $S \leftarrow \textsf{Merge}_{\eps'}(S, G_i)$ and $W_C \pluseq W_i$.  
  \IF {($W_C/\hat W > 1+\eps/2$)}
    \STATE Update $\hat W \leftarrow W_C$, and broadcast $\hat W$ to all sites.  
  \ENDIF
\end{algorithmic}
\end{algorithm}
\vspace{-3mm}
\begin{lemma}
   (P1) Algorithms \ref{alg:P1w-site} and \ref{alg:P1w-coord} maintain that
  for any item $e \in [u]$ that $|f_e(S) - f_e(A)| \leq \eps W_A$.
  The total communication cost is $O((m/\eps^2) \log (\beta N))$
  elements.
\end{lemma}
\begin{proof}
  For any item $e \in [u]$, the coordinator's summary $S$ has error
  coming from two sources.  First is the error as a result of merging
  all summaries sent by each site.  By running these with an error
  parameter $\eps' = \eps/2$, we can guarantee~\cite{ACHPWY12} that
  this leads to at most $\eps' W_C \leq \eps W_A/2$, where $W_C$ is
  the weight represented by all summaries sent to the coordinator,
  hence less than the total weight $W_A$.

  The second source is all elements on the sites not yet sent to the
  coordinator.  Since we guarantee that each site has total weight at
  most $\tau = (\eps/2m) \hat W \leq (\eps/2m) W$, then that is also
  an upper bound on the weight of any element on each site.  Summing
  over all sites, we have that the total weight of any element not
  communicated to the coordinator is at most $m \cdot (\eps/2m) W =
  (\eps/2) W$.

  Combining these two sources of error implies the total error on each
  element's count is \emph{always} at most $\eps W$, as desired.

  The total communication bound can be seen as follows.  Each message
  takes $O(1/\eps)$ space.  The coordinator sends out a message to all
  $m$ sites every (at most) $m$ updates it sees from the coordinators;
  call this period an epoch.  Thus each epoch uses $O(m/\eps)$
  communication.  In each epoch, the size of $W_C$ (and hence $\hat
  W$) increases by an additive $m \cdot (\eps/2m) \hat W \geq (\eps/4)
  W_A$, which is at least a relative factor $(1+\eps/4)$.  Thus
  starting from a weight of $1$, there are $k$ epochs until
  $1\cdot(1+\eps/4)^k \geq \beta N$, and thus $k = O(\frac{1}{\eps}
  \log (\beta N))$.  So after all $k$ epochs the total communication
  is at most $O((m/\eps^2) \log(\beta N))$.
\end{proof}

\subsection{Weighted Heavy-Hitters Protocol 2}
\label{sec:P2w}
Next we observe that we can significantly improve the communication
cost of protocol P1 (above) using an observation, based on an
unweighted frequency estimation protocol by Yi and
Zhang~\cite{yi2013optimal}.  Algorithms \ref{alg:P2w-site} and
\ref{alg:P2w-coord} summarize this protocol.

Each site takes an approach similar to Algorithm \ref{alg:P1w-site},
except that when the weight threshold is reached, it does not send the
entire summary it has, but only the weight at the site.  It still
needs to report heavy elements, so it also sends $e$ whenever any
element $e$'s weight has increased by more than $(\eps/m) \hat W$
since the last time information was sent for $e$.  Note here it only
sends that element, not all elements.
 
After the coordinator has received $m$ messages, then the total weight
constraint must have been violated.  Since $W \leq \beta N$, at most
$O(\log_{(1+\eps)}(\beta N)) = O((1/\eps)\log(\beta N))$ rounds are
possible, and each round requires $O(m)$ total weight messages.
It is a little trickier (but not too hard) to see it requires only a total
of $O((m/\eps)\log(\beta N))$ element messages, as follows from the
next lemma; it is in general not true that there are $O(m)$ such
messages in one round.

\begin{algorithm}
\caption{\label{alg:P2w-site} P2: Tracking heavy-hitters (at site $S_i$)}
\begin{algorithmic}
  \FOR{each item $(a_n, w_n)$}
  \STATE $W_i \pluseq w_n$ and $\Delta_{a_n} \pluseq w_n$.  
  \IF {($W_i \geq (\eps/m) \hat{W}$)} 
  	\STATE Send $(\textsf{total},W_i)$ to $C$ and reset $W_i=0$.  
  \ENDIF
  \IF {($\Delta_{a_n} \geq (\eps/m) \hat{W}$)}
  	\STATE Send $(a_n,\Delta_{a_n})$ to $C$ and reset $\Delta_{a_n}=0$.  
  \ENDIF
  \ENDFOR
\end{algorithmic}
\end{algorithm}
\vspace{-1mm}
\begin{algorithm}
\caption{\label{alg:P2w-coord} P2: Tracking heavy-hitters (at $C$)}
\begin{algorithmic}
  \STATE On message $(\textsf{total}, W_i)$:
  \STATE Set $\hat{W} \pluseq W_i$ and $\#\s{msg} \pluseq 1$.    
  \IF {($\#\s{msg} \geq m$)}
    \STATE Set $\#\s{msg}=0$ and broadcast $\hat W$ to all sites.  
  \ENDIF
  \STATE On message $(a_n, \Delta_n)$: set $\hat{W}_{a_n} \pluseq \Delta_{a_n}$.  
\end{algorithmic}
\end{algorithm}

\vspace{-4mm}
\begin{lemma}
\label{lem:P2w-comm}
After $r$ rounds, at most $O(m \cdot r)$ element update messages have been sent.
\end{lemma}
\begin{proof}
  We prove this inductively.  Each round gets a budget of $m$
  messages, but only uses $t_i$ messages in round $i$.  We maintain a
  value $T_r = r \cdot m - \sum_{i = 1}^r t_i$.  We show inductively
  that $T_r \geq 0$ at all times.

  The base case is clear, since there are at most $m$ messages in
  round $1$, so $t_1 \leq m$, thus $T_1 = m - t_1 \geq 0$.  Then since
  it takes less than $1$ message in round $i$ to account for the
  weight of a message in a round $i' < i$.  Thus, if
  $\sum_{i=1}^{r-1} t_i = n_r$, so $k_r = (r-1)m - n_r$, then if round
  $i$ had more than $m + k_r$ messages, the coordinator would have
  weight larger than having $m$ messages from each round, and it would
  have at some earlier point ended round $r$.  Thus this cannot
  happen, and the inductive case is proved.
\end{proof}

The error bounds follow directly from the unweighted case from
\cite{yi2013optimal}, and is similar to that for (P1).  We can thus
state the following theorem.

\vspace{-3mm}
\begin{theorem}
  Protocol 2 (P2) sends $O(\frac{m}{\eps} \log(\beta N))$ total
  messages, and approximates all frequencies within $\eps W$.
\end{theorem}
\vspace{-2mm}

One can use the space-saving algorithm~\cite{metwally2006integrated}
to reduce the space on each site to $O(m/\eps)$, and the space on the
coordinator to $O(1/\eps)$.

\subsection{Weighted Heavy-Hitters Protocol 3}
\label{sec:P3w}
The next protocol, labeled (P3), simply samples elements to send to
the coordinator, proportional to their weight.

Specifically we combine ideas from priority sampling~\cite{DLT07} for
without replacement weighted sampling, and distributed sampling on
unweighted elements \cite{CMYZ12}.  In total we maintain a random
sample $S$ of size at least $s = O(\frac{1}{\eps^2}
\log\frac{1}{\eps})$ on the coordinator, where the elements are chosen
{\em proportional to their weights}, unless the weights are large
enough (say greater than $W/s$), in which case they are always chosen.
By deterministically sending all large enough weighted elements, not
only do we reduce the variance of the approach, but it also means the
protocol naturally sends the full dataset if the desired sample size
$s$ is large enough, such as at the beginning of the stream.
Algorithm \ref{alg:P3w-site} and Algorithm \ref{alg:P3w-coord}
summarize the protocol.

We denote total weight of sample by $W_S$.  
On receiving a pair $(a_n, w_n)$, a site generates a random number
$r_n \in \unif(0,1)$ and assigns a priority $\rho_n = w_n / r_n$ to
$a_n$. Then the site sends triple $(a_n, w_n, \rho_n)$ to the
coordinator if $\rho_n \geq \tau$, where $\tau$ is a global threshold
provided by the coordinator.

Initially $\tau$ is $1$, so sites simply send any items they receive
to the coordinator. At the beginning of further rounds, the
coordinator doubles $\tau$ and broadcasts it to all sites. Therefore
at round $j$, $\tau = \tau_j = 2^j$. In any round $j$, the coordinator
maintains two priority queues $Q_{j}$ and $Q_{j+1}$. On receiving a
new tuple $(a_n, w_n, \rho_n)$ sent by a site, the coordinator places
it into $Q_{j+1}$ if $\rho_n \geq 2\tau$, otherwise it places $a_n$
into $Q_j$.

Once $|Q_{j+1}| = s$, the round ends. At this time, the coordinator
doubles $\tau$ as $\tau = \tau_{j+1} = 2\tau_j$ and broadcasts it to
all sites.  Then it discards $Q_j$ and examines each item
$(a_n,w_n,\rho_n)$ in $Q_{j+1}$, if $\rho_n \geq 2\tau$, it goes into
$Q_{j+2}$, otherwise it remains in $Q_{j+1}$.

\begin{algorithm}
\caption{\label{alg:P3w-site} P3: Tracking heavy-hitters (at site $S_i$)}
\begin{algorithmic}
  \FOR{$(a_n, w_n)$ in round $j$}
  \STATE choose $r_n \in \unif(0,1)$ and set $\rho_n = w_n/r_n$.
  \STATE \textbf{if} $\rho_n \geq \tau$ \textbf{then} send $(a_n, w_n, \rho_n)$ to $C$.  
  \ENDFOR
\end{algorithmic}
\end{algorithm}

\begin{algorithm}
\caption{\label{alg:P3w-coord} P3: Tracking heavy-hitters (at $C$)}
\begin{algorithmic}
  \STATE On input of $(a_n, w_n, \rho_n)$ from any site in round $j$:
  \STATE \textbf{if} $\rho > 2 \tau_j$ \textbf{then} put $a_n$ in $Q_{j+1}$,
  \STATE \hspace{14mm} \textbf{else} put $a_n$ in $Q_j$.  
  \IF {$|Q_{j+1}| \geq s$}      
    \STATE Set $\tau_{j+1} = 2 \tau_j$; broadcast $\tau_{j+1}$ to all sites.  
    \FOR {$(a_n,w_n, \rho_n) \in Q_{j+1}$} 
    \STATE \textbf{if} $\rho_n > 2\tau_{j+1}$, put $a_n$ in $Q_{j+2}$.  
    \ENDFOR
  \ENDIF
\end{algorithmic}
\end{algorithm}

At any time, a sample of size exactly $s$ can be derived by
subsampling from $Q_j \cup Q_{j+1}$.  But it is preferable to use a
larger sample $S = Q_j \cup Q_{j+1}$ to estimate properties of $A$, so
we always use this full sample.

\Paragraph{Communication analysis.}
The number of messages sent to the coordinator in each round is $O(s)$
with high probability. To see that, consider an arbitrary round
$j$. Any item $a_n$ being sent to coordinator at this round, has
$\rho_n \geq \tau$. This item will be added to $Q_{j+1}$ with
probability
\begin{align*}
\Pr(\rho_n \geq 2\tau \mid \rho_n \geq \tau) 
&= 
\frac{\Pr(\rho_n \geq 2\tau)}{\Pr(\rho_n \geq \tau)} 
= 
\frac{\Pr(r_n \leq \frac{w_n}{2\tau})}{\Pr(r_n \leq \frac{w_n}{\tau})} 
\\ &= 
\frac{\min(1,\frac{w_n}{2\tau})}{\min(1,\frac{w_n}{\tau})}
\geq 
\frac{1}{2}.
\end{align*}
Thus sending $4s$ items to coordinator, the expected number of items
in $Q_{j+1}$ would be greater than or equal to $2s$. Using a
Chernoff-Hoeffding bound $\Pr(2s-|Q_{j+1}|>s) \leq \exp(-2s^2/4s) =
\exp(-s/2)$. So if in each round $4s$ items are sent to coordinator,
with high probability (at least $1-\exp(-s/2)$), there would be $s$
elements in $Q_{j+1}$. Hence each round has $O(s)$ items sent with
high probability.
The next lemma, whose proof is 
contained in the appendix, 
bounds the number of rounds. Intuitively, each round requires the
total weight of the stream to double, starting at weight $s$, and this
can happen $O(\log(\beta N/s))$ times.  \vspace{-2mm}
\begin{lemma}
\label{LEM:P3W-ROUNDS}
The number of rounds is at most $O(\log(\beta N /s))$ with probability
at least $1-e^{-\Omega(s)}$.
\end{lemma}

Since with probability at least $1-e^{-\Omega(s)}$, in each
round the coordinator receives $O(s)$ messages from all sites
and broadcasts the threshold to all $m$ sites, we can then combine
with Lemma \ref{LEM:P3W-ROUNDS} to bound the total messages.
\vspace{-2mm}
\begin{lemma}
\label{lem:P3w-mes}
This protocol sends $O((m+s) \log\frac{\beta N}{s})$ messages with
probability at least $1-e^{-\Omega(s)}$.  We set $s =
\Theta(\frac{1}{\eps^2} \log \frac{1}{\eps})$.
\end{lemma}
\vspace{-2mm}

Note that each site only requires $O(1)$ space to store the threshold,
and the coordinator only requires $O(s)$ space.

\Paragraph{Creating estimates.}
To estimate $f_e(A)$ at the coordinator, we use a set $S' = Q_j \cup
Q_{j+1}$ which is of size $|S'| = s' > s$.  Let $\hat \rho$ be the
priority of the smallest priority element in $S'$.  Let $S$ be all
elements in $S'$ except for this single smallest priority element.
For each of the $s'-1$ elements in $S$ assign them a weight $\bar w_i
= \max(w_i, \hat \rho)$, and we set $W_S = \sum_{a_i \in S} \bar w_i$.
Then via known priority sampling results~\cite{DLT07,Sze06}, it
follows that $\E[W_S] = W_A$ and that $(1-\eps) W_A \leq W_S \leq
(1+\eps) W_A$ with large probability (say with probability $1-\eps^2$,
based on variance bound $\Var[W_S] \leq W_A^2/(s'-2)$~\cite{Sze06} and
a Chebyshev bound).  Define $S_e = \{a_n \in S \mid a_n = e\}$ and
$f_e(S) = \sum_{a_n \in S_e} \bar w_n$.

\vspace{2mm} The following lemma, whose proof is in the appendix,
shows that the sample
maintained at the coordinator gives a good estimate on item
frequencies.  At a high-level, we use a special Chernoff-Hoeffding
bound for negatively correlated random variables~\cite{PS97} (since
the samples are without replacement), and then only need to consider
the points selected that have small weights, and thus have values in
$\{0, \hat \rho\}$.  \vspace{-2mm}
\begin{lemma} 
\label{LEM:P3W-EST}
With $s = \Theta((1/\eps^2) \log(1/\eps))$, the coordinator can use the
estimate from the sample $S$ such that, with large probability, for
each item $e \in [u]$, $\left|f_e(S) - f_e(A) \right| \leq \eps W_A$.
\end{lemma}

\vspace{-4mm}
\begin{theorem}
\label{thm:P3w-mes}
Protocol 3 (P3) sends $O((m+s) \log\frac{\beta N}{s})$ messages with large
probability; It gets a set $S$ of size $s =\Theta(\frac{1}{\eps^2} \log
\frac{1}{\eps})$ so that $|f_e(S) - f_e(A)| \leq \eps W$.
\end{theorem}

\subsubsection{Sampling With Replacement}
\label{sec:P3w-SwR}
We can show similar results on $s$ samples \emph{with replacement},
using $s$ independent samplers. In round $j$, for each element
$(a_n,w_n)$ arriving at a local site, the site generates $s$
independent $r_n$ values, and thus $s$ priorities $\rho_n$. If any of
them is larger than $\tau_j$, then the site forwards it to
coordinator, along with the index (or indices) of success.

For each of $s$ independent samplers, say for sampler $t \in [s]$, the
coordinator maintains the top $2$ priorities $\rho^{(1)}_t$ and
$\rho^{(2)}_t$, $\rho^{(1)}_t > \rho^{(2)}_t$, among all it received.
It also keeps the element information $a_t$ associated with
$\rho^{(1)}$.  For the sampler $i \in [s]$, the coordinator keeps a
weight $\bar w_i = \rho^{(2)}_i$.  One can show that $\E[\bar w_i] =
W$, the total weight of the entire stream~\cite{DLT07}.  We improve
the global estimate as $\hat W = (1/s) \sum_{i=1}^s \bar w_i$, and
then assign each element $a_i$ the same weight $\hat w_i = \hat W/s$.
Now $\E[\sum_{i=1}^s \hat w_i] = W$, and each $a_i$ is an independent
sample (with replacement) chosen proportional to its weight.  Then
setting $s = O((1/\eps^2) \log (1/\eps))$ it is known that these
samples can be used to estimate all heavy hitters within $\eps W$ with
probability at least $1-e^{-\Omega(s)}$.

The $j$th round terminates when the $\rho^{(2)}_i$ for all $i$ is
larger than $2\tau_j$. At this point, coordinator sets $\tau_{j+1} = 2
\tau_j$, informs all sites of the new threshold and begins the $(j+1)$th
round.

\Paragraph{Communication analysis.} 
Since this protocol is an adaptation of existing
results~\cite{CMYZ12}, its communication is $O((m + s \log s) \log
(\beta N) = O((m + \frac{1}{\eps^2}\log^2 \frac{1}{\eps}) \log (\beta
N))$
messages.
This result doesn't improve the error bounds or communication bounds
with respect to the without replacement sampler described above, as is
confirmed in Section \ref{sec:exp}.  Also in terms of running time (without
parallelism at each site), sampling {\em without replacement} will be
better. 

\subsection{Weighted Heavy-Hitters Protocol 4}
\label{sec:P4w}
This protocol is inspired by the unweighted case from Huang
\etal~\cite{huang2012randomized}.  Each site maintains an estimate of
the total weight $\hat W$ that is provided by the coordinator and
always satisfies $\hat W \leq W \leq 2 \hat W$, with high
probability.  It then sets a probability $p = 2\sqrt{m}/(\eps \hat W)$.  Now
given a new element $(a,w)$ with some probability $\bar p$, it sends
to the coordinator $(e, \bar w_{e,j} = f_e(A_j))$ for $a=e \in [u]$;
this is the total weight of all items in its stream that equal element
$e$.  Finally the coordinator needs to adjust each $\bar w_{e,j}$ by
adding $1/p-1$ (for elements that have been witnessed) since that
is the expected number of items with element $e$ in the stream until 
the next update for $e$.

If $w$ is an integer, then one option is to pretend it is actually $w$
distinct elements with weight $1$.  For each of the $w$ elements we
create a random variable $Z_i$ that is $1$ with probability $p$ and
$0$ otherwise.  If \emph{any} $Z_i =1$, then we send $f_e(A_j)$.
However this is inefficient (say if $w = \beta = 1000$), and only
works with integer weights.

Instead we notice that at least one $Z_i$ is $1$ if none are $0$, with
probability $1 - (1-p)^w \approx 1 - e^{-pw}$.  So in the integer
case, we can set $\bar p = 1-(1-p)^w$, and then since we send a more
accurate estimate of $f_e$ (as it essentially comes later in the
stream) we can directly apply the analysis from Huang
\etal~\cite{huang2012randomized}.  To deal with non integer weights,
we set $\bar p = 1-e^{-p w}$, and describe the approach formally on a
site in Algorithm \ref{alg:P4w-site}.

Notice that the probability of sending an item is asymptotically the
same in the case that $w=1$, and it is smaller otherwise (since we
send at most one update $\bar w_{e,j}$ per batch).  Hence the
communication bound is asymptotically the same, except for the number
of rounds.  Since the weight is broadcast to the sites from the
coordinator whenever it doubles, and now the total weight can be
$\beta N$ instead of $N$, the number of rounds is $O(\log (\beta N))$
and the total communication is $O((\sqrt{m}/\eps) \log (\beta N))$
with high probability.

\begin{algorithm}
\caption{\label{alg:P4w-site} P4: Tracking of heavy-hitters (at site $S_j$)}
\begin{algorithmic}
  \STATE Given weight $\hat W$ from $C$, set $p = 2\sqrt{m}/(\eps \hat W)$.

  \FOR {each item $(a,w)$ it receives} 
  \STATE For $a = e$ update $f_e(A_j) := f_e(A_j) + w$.  
  \STATE Set $\bar p = 1 - e^{-pw}$.
  \STATE With probability $\bar p$ send $\bar w_{e,j} = f_e(A_j)$ to $C$.
\ENDFOR
\end{algorithmic}
\end{algorithm}

When the coordinator is sent an estimate $\bar w_{e,j}$ of the total
weight of element $e$ at site $j$, it needs to update this estimate
slightly as in Huang \etal, so that it has the right expected value.
It sets $\hat w_{e,j} = \bar w_{e,j} + 1/p$, where again $p = 2
\sqrt{m}/(\eps \hat W)$; $\hat w_{e,j} = 0$ if no such
messages are sent.  The coordinator then estimates each $f_e(A)$ as
$\hat W_e = \sum_{j=1}^m \hat w_e$.

We first provide intuition how the analysis works, if we used $\bar p
= 1 - (1-p)^w$ (i.e. $\approx 1 - e^{-pw}$) and $w$ is an integer.  In this
case, we can consider simulating the process with $w$ items of weight
$1$; then it is identical to the unweighted algorithm, except we
always send $\bar w_{e,j}$ at then end of the batch of $w$ items.
This means the expected number until the next update is
still $1/p-1$, and the variance of $1/p^2$ and error bounds of Huang
\etal~\cite{huang2012randomized} still hold.

\vspace{-2mm}
\begin{lemma}
  The above protocol guarantees that $|f_e(A) - \hat W_e | \leq
  \eps W$ on the coordinator, with probability at least $0.75$.
\label{lem:P4w-err}
\end{lemma}
\begin{proof}
  Consider a value of $k$ large enough so that $w \cdot 10^k$ is
  always an integer (i.e., the precision of $w$ in a system is at most
  $k$ digits past decimal point).  Then we can hypothetically simulate
  the unweighted case using $w_k = w \cdot 10^k$ points.  Since now
  $\hat W$ represents $10^k$ times as many unweighted elements, we
  have $p_k = p/10^k = \sqrt{m}/(\eps \hat W 10^k)$.  This means
  the probability we send an update should be $1-(1-p_k)^{w_k}$ in
  this setting.

  Now use that for any $x$ that $\lim_{n \to \infty} (1-\frac{x}{n})^n
  = e^{-x}$.  Thus setting $n = w_k$ and $x = p_k \cdot w_k = (p/10^k)
  (w 10^k) = pw$ we have $\lim_{k \to \infty} 1 - (1-p_k)^{w_k} =
  1-e^{-pw}$.

  Next we need to see how this simulated process affects the error on
  the coordinator.  Using results from Huang \etal
  \cite{huang2012randomized}, where they send an estimate $\bar
  w_{e,j}$, the expected value $\E[\bar w_{e,j}] = f_e(A_j) - 1/p + 1$
  at any point afterwards where that was the last update.  This
  estimates the count of weight $1$ objects, so in the case where they
  are weight $10^{-k}$ objects the estimate of $f_e(A_j)^{(k)} =
  f_e(A_j) 10^{k}$ is using $\bar w_{e,j}^{(k)} = \bar w_{e,j}
  10^{k}$.  Then, in the limiting case (as $k \to \infty$), we adjust
  the weights as follows.
\begin{align*}
\E[\bar w_{e,j}] 
& = 
\E[\bar w_{e,j}^{(k)}] 10^{-k} 
= 
(f_e(A_j)^{(k)} - 1/p_k + 1) \cdot 10^{-k}
\\ &= 
(f_e(A_j) 10^k - \frac{10^k}{p} + 1) 10^{-k}
=
f_e(A_j) - \frac{1}{p} + 10^{-k},
\end{align*}
so as $\lim_{k \to \infty} \E[\bar w_{e,j}] = f_e(A_j) - 1/p$.  So our
procedure has the right expected value.  Furthermore, it also follows
that the variance is still $1/p^2$, and thus the error bound from
\cite{huang2012randomized} that any $|f_e(A) - \hat W_e| \leq \eps W$
with probability at least $0.75$ still holds.
\end{proof}

\vspace{-5mm}
\begin{theorem}
Protocol 4 (P4) sends $O(\frac{\sqrt{m}}{\eps} \log (\beta N))$ total messages and with probability $0.75$ has $|f_e(A) - \hat W_e| \leq \eps W$.    
\end{theorem}
\vspace{-2mm}

The bound can be made to hold with probability $1-\delta$ by running
$\log(2/\delta)$ copies and taking the median.  The space on each site
can be reduced to $O(1/\eps)$ by using a weighted variant of the
space-saving algorithm~\cite{metwally2006integrated}; the space on the
coordinator can be made $O(m/\eps)$ by just keeping weights for which
$\bar w_{i,e} \geq 2 \eps \hat W_j$, where $\hat W_j$ is a
$2$-approximation of the weight on site $j$.

\section{Distributed Matrix Tracking}
\label{sec:matrix}

We will next extend weighted frequent item estimation protocols to
solve the problem of tracking an approximation to a distributed
matrix. Each element $a_n$ of a stream for any site is now a row of
the matrix. As we will show soon in our analysis, it will be
convenient to associate a weight with each element defined as the
squared norm of the row, i.e., $w_n=\|a_n\|^2$. Hence, for reasons
outlined in Section \ref{sec:weights}, we assume in our analysis that the squared norm of
every row is bounded by a value $\beta$.  There is a
designated coordinator $C$ who has a two-way communication channel
with each site and whose job is to maintain a much smaller matrix $B$
as an approximation to $A$ such that for any unit vector $x\in \R^{d
  \times 1}$ (with $\|x\|=1$) we can ensure that:
\[
| \|A x\|^2 - \| B x\|^2 | \leq \eps \| A \| _F^2.
\]
Note that the covariance of $A$ is captured as $A^T A$ (where ${}^T$
represents a matrix transpose), and that the above expression is
equivalent to
\[
\|A^T A - B^T B\|_2 \leq \eps \|A\|_F^2.
\]
Thus, the approximation guarantee we preserve shows that the
covariance of $A$ is well-approximated by $B$.  And the covariance is
the critical property of a matrix that needs to be (approximately)
preserved as the basis for most downstream data analysis; e.g., for
PCA or LSI.

Our measures of complexity will be the communication cost and the space
used at each site to process the stream.  We measure communication in
terms of the number of messages, where each message is a row of length $d$, the same as the input stream.  Clearly, the space and
computational cost at each site and coordinator is also important, but
since we show that all proposed protocols can be run as streaming
algorithms at each site, and will thus not be space or computation
intensive.

\Paragraph{Overview of protocols.} 
The protocols for matrix tracking mirror those of weighted item
frequency tracking.  This starts with a similar batched streaming
baseline P1. Protocol P2 again reduces the total communication bound,
where a global threshold is given for each ``direction'' instead of
the total squared Frobenious norm.  Both P1 and P2 are deterministic.
Then matrix tracking protocol P3 randomly selects rows with
probability proportional to their squared norm and maintains an
$\eps$-sample at the coordinator. Using this sample set, we can derive
a good approximation.

Given the success of protocols P1, P2, and P3, it is tempting to also
extend protocol P4 for item frequency tracking in Section
\ref{sec:P4w} to distributed matrix tracking.  However, unlike the
other protocols, we can show that the approach described in Algorithm
\ref{alg:P4w-site} cannot be extended to matrices in any
straightforward way while still maintaining the same communication
advantages it has (in theory) for the weighted heavy-hitters case.
Due to lack of space, we defer this explanation, and the related experimental
results to the appendix section.

\subsection{Distributed Matrix Tracking Protocol 1}\
\label{sec:P1m}
We again begin with a batched version of a streaming algorithm, shown
as Algorithm \ref{alg:P1m-site} and \ref{alg:P1m-coord}.  That is we
run a streaming algorithm (e.g. Frequent Directions~\cite{Lib12},
labeled \textsf{FD}, with error $\eps' = \eps/2$) on each site, and
periodically send the contents of the memory to the coordinator.
Again this is triggered when the total weight (in this case squared
norm) has increased by $(\eps/2m) W$.

\begin{algorithm}
\caption{\label{alg:P1m-site} P1: Deterministic Matrix Tracking  (at $S_i$)}
\begin{algorithmic}
  \FOR{$(a_n, w_n)$ in round $j$}
    \STATE Update $B_i \leftarrow \textsf{FD}_{\eps'}(B_i,a_n)$; and $F_i \pluseq \|a_n\|^2$.    
    \IF {($F_i \geq \tau = (\eps/2m) \hat F$)}
      \STATE Send $(B_i,F_i)$ to coordinator; make $B_i,F_i$ empty.  
    \ENDIF
  \ENDFOR
\end{algorithmic}
\end{algorithm}

\begin{algorithm}
\caption{\label{alg:P1m-coord} P1: Deterministic Matrix Tracking (at $C$)}
\begin{algorithmic}
  \STATE On input $(B_i,F_i)$:
  \STATE Update sketch $B \leftarrow \textsf{Merge}_{\eps'}(B, B_i)$ and $F_C \pluseq F_i$.  
  \IF {($F_C/\hat F > 1+\eps/2$)}
    \STATE Update $\hat F \leftarrow F_C$, and broadcast $\hat F$ to all sites.  
  \ENDIF
\end{algorithmic}
\end{algorithm}

As with the similar frequency tracking algorithm, based on Frequent
Directions~\cite{Lib12} satisfying the mergeable
property~\cite{ACHPWY12}, we can show this maintains at most $\eps
\|A\|_F^2$ total error at all times, and requires a total of
$O((m/\eps^2) \log (\beta N))$ total rows of communication.

\subsection{Distributed Matrix Tracking Protocol 2} 
\label{sec:P2m}
Again, this protocol is based very closely on a weighted heavy-hitters
protocol, this time the one from Section \ref{sec:P2w}.

Each site $S_j$ maintains a matrix $B_j$ of the rows seen so far at
this site and not sent to coordinator.  In addition, it maintains
$\hat F$, an estimate of $\|A\|_F^2$, and $F_j = \|B_j\|_F^2$,
denoting the total squared Frobenius norm received since its last
communication to $C$ about $\hat F$.  The coordinator $C$ maintains a
matrix $B$ approximating $A$, and $\hat{F}$, an
$\eps$-approximation of $\|A\|_F^2$.

Initially each $\hat F$ is set to zero for all sites.  When site $j$
receives a new row, it calls Algorithm \ref{alg:det-tracking-site},
which basically sends $\|B_j x\|^2$ in direction $x$ when it is
greater than some threshold provided by the coordinator, if one
exists.

\begin{algorithm}
  \caption{\label{alg:det-tracking-site} P2: Deterministic Matrix Tracking  (at $S_j$)}
\begin{algorithmic}
  \STATE $F_j \pluseq \| a_i\|^2$
  \IF {($F_j \geq \frac{\eps}{m} \hat{F}$)}
    \STATE Send $F_j$ to coordinator; set $F_j =0$.
  \ENDIF
  \STATE  Set $B_j \leftarrow [B_j ; a_i]$
  \STATE  $[U, \Sigma, V] = \svd(B_j)$
  \FOR {($(v_\ell, \sigma_\ell)$ such that $\sigma_\ell^2 \geq \frac{\eps}{m} \hat{F}$)}
    \STATE Send $\sigma_\ell v_\ell$ to coordinator; set $\sigma_\ell = 0$.
  \ENDFOR
  \STATE $B_j = U \Sigma V^T$
\end{algorithmic}
\end{algorithm}
\vspace{-1mm}
\begin{algorithm}
  \caption{\label{alg:det-tracking-coord} P2: Deterministic Matrix Tracking  (at $C$)}
\begin{algorithmic}
	\STATE On a scalar message $F_j$ from site $S_j$
	\STATE Set $\hat F \pluseq F_j$ and $\#\s{msg} \pluseq 1$.  
	\IF {($\#\s{msg} \geq m$)}
	\STATE Set $\#\s{msg}=0$ and broadcast $\hat F$ to all sites.  
	\ENDIF
	\STATE On a vector message $r = \sigma v$: append $B \leftarrow [B; r]$
\end{algorithmic}
\end{algorithm}

On the coordinator side, it either receives a vector form message
$\sigma v$, or a scalar message $F_j$.  For a scalar $F_j$,
it adds it to $\hat F$.  After at most $m$ such scalar messages, it
broadcasts $\hat F$ to all sites.  For vector message $r = \sigma
v$, the coordinator updates $B$ by appending $r$ to $B \leftarrow
[B; r]$. The coordinator's protocol is summarized in Algorithm
\ref{alg:det-tracking-coord}.

\vspace{-3mm}
\begin{lemma}
  At all times the coordinator maintains $B$ such that for any unit vector $x$
\begin{equation}
\label{eq:invariant2}
\| A x\|^2 - \eps \| A\|_F^2 \leq \| B x\|^2 \leq \| A x\|^2
\end{equation}
\end{lemma}
\begin{proof}
  To prove this, we also need to show it maintains another property on
  the total squared Frobenious norm:
\begin{equation}
\label{eq:invariant1}
(1-2\eps)\| A\|_F^2  < \hat{F} \leq \| A\|_F^2.  
\end{equation}
This follows from the analysis in Section \ref{sec:P2w} since the
squared Frobenius norm is additive, just like weights.  The following
analysis for the full lemma is also similar, but requires more care in
dealing with matrices. First, for any $x$ we have
\[
\|A x\|^2 = \|B x\|^2 + \sum_{j=1}^m \|B_j x\|^2.
\]
This follows since $\|A x\|^2 = \sum_{i=1}^n \langle a_i, x\rangle^2$,
so if nothing is sent to the coordinator, the sum can be decomposed
like this with $B$ empty.  We just need to show the sum is preserved
when a message $r = \sigma_1 v_1$ is sent.  Because of the
orthogonal decomposition of $B_j$ by the $\svd(B_j) = [U, \Sigma, V]$,
then $\|B_j x\|^2 = \sum_{\ell=1}^d \langle \sigma_\ell v_\ell, x
\rangle^2$.  Thus if we send any $\sigma_\ell v_\ell$ to the
coordinator, append it to $B$, and remove it from $B_j$, the sum is
also preserved.  Thus, since the norm on $B$ is always less than on
$A$, the right side of (\ref{eq:invariant2}) is proven.

To see the left side of (\ref{eq:invariant2}) we need to use
(\ref{eq:invariant1}), and show that not too much mass remains on the
sites.  First we bound $\|B_j x\|^2$.
\[
\|B_j x\|^2 
= 
\sum_{\ell=1}^d \sigma_\ell^2 \langle v_\ell, x\rangle^2 
\leq 
\sum_{\ell=1}^d \frac{\eps}{m} \hat F \langle v_\ell, x\rangle^2 
=
\frac{\eps}{m} \hat F
\leq 
\frac{\eps}{m} \|A\|_F^2.
\]
And thus $\sum_{j=1}^m \|B_j x\|^2 \leq m \frac{\eps}{m} \|A\|_F^2 = \eps \|A\|_F^2$ and
hence 
\[
\|A x\|^2 \leq \|B x\|^2 + \sum_{j=1}^m \|B_j x\|^2 \leq \|B x\|^2 + \eps \|A\|_F^2. \qed
\]
\end{proof}\vspace{-2mm}

The communication bound follows directly from the analysis of the
weighted heavy hitters since the protocols for sending messages and
starting new rounds are identical with $\|A\|_F^2$ in place of $W$,
and with the squared norm change along the largest direction (the top
right singular value) replacing the weight change for a single
element.  Thus the total communication is $O(\frac{m}{\eps} \log
(\beta N))$.

\vspace{-2mm}
\begin{theorem}
  For a distributed matrix $A$ whose squared norm of rows are bounded
  by $\beta$ and for any $0 \leq \eps \leq 1$, the
  above protocol (P2) continuously maintains $\hat{A}$ such that $0 \leq \|
  A x\|^2 - \| B x\|^2 \leq \eps \| A\|_F^2$ and incurs a total
  communication cost of $O((m/\eps)\log (\beta N))$ messages.
\end{theorem}

\vspace{-2mm}
\Paragraph{Bounding space at sites.}
It is possible to also run a small space streaming algorithm on each
site $j$, and also maintain the same guarantees.  The Frequent
Directions algorithm~\cite{Lib12}, presented a stream of rows $a_i$
forming a matrix $A$, maintains a matrix $\tilde A$ using $O(1/\eps')$
rows such that $0 \leq \|Ax\|^2 - \|\tilde A x\|^2 \leq \eps' \|
A\|_F^2$ for any unit vector $x$.

In our setting we run this on two matrices on each site with $\eps' =
\eps / 4m$.  (It can actually just be run on $B_j$, but then the proof
is much less self-contained.)  It is run on $A_j$, the full matrix.
Then instead of maintaining $B_j$ that is $A_j$ after subtracting all
rows sent to the coordinator, we maintain a second matrix $S_j$ that
contains all rows sent to the coordinator; it appends them one by one,
just as in a stream.  Now $\|B_j x\|^2 = \|A_jx\|^2 - \|S_jx\|^2$.
Thus if we replace both $A_j$ with $\tilde A_j$ and $S_j$ with $\tilde
S_j$, then we have
\[
\|B_j x\|^2 = \|A_jx \|^2 - \|S_j x\|^2 \leq \|\tilde A_j x\|^2 -
\|\tilde S_j x\|^2 + \frac{\eps}{4m} \|A_j\|_F^2,
\]
and similarly $\|B_j x\|^2 \geq \|\tilde A_j x\|^2 - \|\tilde S_j
x\|^2 - \frac{\eps}{4m} \|A_j\|_F^2$ (since $\|S_j\|_F^2 \leq
\|A_j\|_F^2$).  
From here we will abuse notation and write $\|\tilde
B_j x\|^2$ to represent $\|\tilde A_j x\|^2 - \|\tilde S_j x\|^2$.

Now we send the top singular vectors $v_\ell$ of $\tilde B_j$ to the
coordinator only if $\|\tilde B_j v_\ell\|^2 \geq \frac{3\eps}{4m}
\hat F$.  Using our derivation, thus we only send a message if $\|B_j
v_\ell\|^2 \geq \frac{\eps}{2m} \|A\|_F^2$, so it only sends at most
twice as many as the original algorithm.  Also if $\|B_j v_\ell\|^2 >
\frac{\eps}{m} \|A\|_F^2$ we always send a message, so we do not
violate the requirements of the error bound.

The space requirement per site is then $O(1/\eps') = O(m/\eps)$ rows.
This also means, as with Frequent Directions~\cite{Lib12}, we can run
Algorithm \ref{alg:det-tracking-site} in batch mode, and only call the
\svd\ operation once every $O(1/\eps')$ rows.

It is straightforward to see the coordinator can also use Frequent
Directions to maintain an approximate sketch, and only keep
$O(1/\eps)$ rows.

\subsection{Distributed Matrix Tracking Protocol 3}
\label{sec:P3m}
Our next approach is very similar to that discussed in Section
\ref{sec:P3w}.  On each site we run Algorithm \ref{alg:P3w-site}, the
only difference is that for an incoming row $a_i$, it treats it as an
element $(a_i, w_i = \|a_i\|^2)$.  The coordinator's communication
pattern is also the same as Algorithm \ref{alg:P3w-coord}, the only
difference is how it interprets the data it receives.

As such, the communication bound follows directly from Section
\ref{sec:P3w}; we need $O((m + (1/\eps^2) \log(1/\eps)) \log (\beta N
\eps))$ messages, and we obtain a set $S$ of at least $s =
\Theta((1/\eps)^2 \log(1/\eps))$ rows chosen proportional to their squared
norms; however if the squared norm is large enough, then it is in the
set $S$ deterministically.  To simplify notation we will say that
there are exactly $s$ rows in $S$.

\Paragraph{Estimation by coordinator.}
The coordinator ``stacks'' the set of rows $\{a_1, \ldots, a_s\}$ to
create an estimate $B = [a_1; \ldots; a_s]$.  We will show that
for any unit vector $x$ that $| \|A x\|^2 - \|B x\|^2 | \leq \eps
\|A\|_F^2$.

If we had instead used the weighted sampling with replacement protocol
from Section \ref{sec:P3w-SwR}, and retrieved $s = O(1/\eps^2)$ rows
of $A$ onto the coordinator (sampled proportionally to $\|a_i\|^2$ and
then rescaled to have the same weight), we could immediately show the
desired bound was achieved using know results on column
sampling~\cite{drineas2006fast2}.  However, as is the case with
weighted heavy-hitters, we can achieve the same error bound for
without replacement sampling in our protocol, and this uses less
communication and running time.

Recall for rows $a_i$ such that $\|a_i\|^2 \geq \hat \rho$, (for a
priority $\hat \rho < 2 \tau$) it keeps them as is; for other rows, it
rescales them so their squared norm is $\hat \rho$.  And $\hat \rho$
is defined so that $\E[\|B\|_F^2] = \|A\|_F^2$, thus $\hat \rho \leq
W/s$.  \vspace{-2mm}
\begin{theorem}
  Protocol 3 (P3) uses $O((m+s)\log(\beta N/s))$ messages of
  communication, with $s = \Theta((1/\eps^2)\log(1/\eps))$, and for
  any unit vector $x$ we have $| \|A x\|^2 - \|B x\|^2 | \leq \eps
  \|A\|_F^2$, with probability at least $1 -1/s$.
\end{theorem}
\begin{proof}
  The error bound roughly follows that of Lemma \ref{LEM:P3W-EST}.  We
  apply the same negatively correlated Chernoff-Hoeffding bound but instead define random
  variable $X_{i,x} = \langle a_i, x\rangle^2$.  Thus $M_x =
  \sum_{i=1}^s X_{i,x} = \|B x\|^2$.  Again $\Delta = \hat \rho$
  (since elements with $\|a_i\|^2 > \hat \rho$ are not random) and
  $\E[M_x] = \|A x\|^2$.  It again follows that
  \begin{equation*} \label{eq:L2bound} 
  \Pr[| \|B x\|^2 - \|A x\|^2 | \leq \eps \|A\|_F^2/2] \leq \exp(-\eps^2 s /32) \leq \delta.
\end{equation*}
Setting $\delta = \Omega(1/s)$ yields that when $s = \Theta((1/\eps^2)
\log(1/\eps))$ this holds with probability at least $1-\delta = 1- 1/s
= 1-1/\Theta((1/\eps)^2\log(1/\eps))$, for any unit vector $x$.
\end{proof}

\noindent
We need $O(1)$ space per site and $O(s)$ space on coordinator.  

\vspace{-2mm}
\section{Experiments}
\label{sec:exp}

\Paragraph{Datasets.}
For tracking the distributed weighted heavy hitters, we generated data
from Zipfian distribution, and set the skew parameter to 2 in order to
get meaningful distributions that produce some heavy hitters per
run. The generated dataset contained $10^7$ points, in order to assign
them weights we fixed the upper bound (default
$\beta=1,000$) and assigned each point a uniform random weight in
range $[1,\beta]$.  Weights are not necessarily integers.

For the distributed matrix tracking problem, we used two large real
datasets ``PAMAP'' and ``YearPredictionMSD'', from the machine
learning repository of UCI. 

PAMAP is a Physical Activity Monitoring dataset and contains data of
18 different physical activities (such as walking, cycling, playing
soccer, etc.), performed by 9 subjects wearing 3 inertial measurement
units and a heart rate monitor. The dataset contains 54 columns
including a timestamp, an activity label (the ground truth) and 52
attributes of raw sensory data. In our experiments, we used a subset
with $N = 629,250$ rows and $d=44$ columns (removing columns
containing missing values), giving a $N \times d$ matrix (when running
to the end). This matrix is low-rank.

YearPredictionMSD is a subset from the ``Million Songs Dataset"
\cite{Bertin-Mahieux2011} and contains the prediction of the release
year of songs from their audio features. It has over 500,000 rows and
$d=90$ columns. We used a subset with $N =300,000$ rows, representing
a $N \times d$ matrix (when running to the end). This matrix has high
rank.

\Paragraph{Metrics.} The efficiency and accuracy of the weighted heavy
hitters protocols are controlled with input parameter $\eps$
specifying desired error tolerance. We compare them on:
\begin{pkl} 
\item \emph{Recall:} The number of true heavy hitters returned by a
  protocol over the correct number of true heavy hitters.
\item \emph{Precision:} The number of true heavy hitters returned by a
  protocol over the total number of heavy hitters returned by the
  protocol. 
\item \s{err}: Average relative error of the frequencies of the
  true heavy hitters returned by a protocol. 
\item \s{msg}: Number of messages sent during a protocol.
\end{pkl}

\vspace{-1mm}
For matrix approximation protocols, we used:
\begin{itemize} \denselist
\item \s{err}: Defined as $\|A^TA - B^TB\|_2/\|A\|_F^2$, where $A$ is
  the input matrix and $B$ is the constructed low rank approximation
  to $A$. It is equivalent to the following: $\max_{\{x,\;\|x\|=1\}}
  (\|A x\|^2 - \|B x\|^2)/\|A\|_F^2$.
\item \s{msg}: Number of messages (scalar-form and vector-form) sent during a
  protocol.
\end{itemize}  \vspace{-1mm}

We observed that both the approximation errors and communication costs
of all methods {\em are very stable with respect to query time}, by
executing estimations at the coordinator at randomly selected time
instances. Hence, we only report the average \s{err} from queries in
the very end of the stream (i.e., results of our methods on really
large streams).

\begin{figure}[t!]
\begin{centering}
\subfigure[recall vs. $\eps$.]{
\includegraphics[width=\figsize]{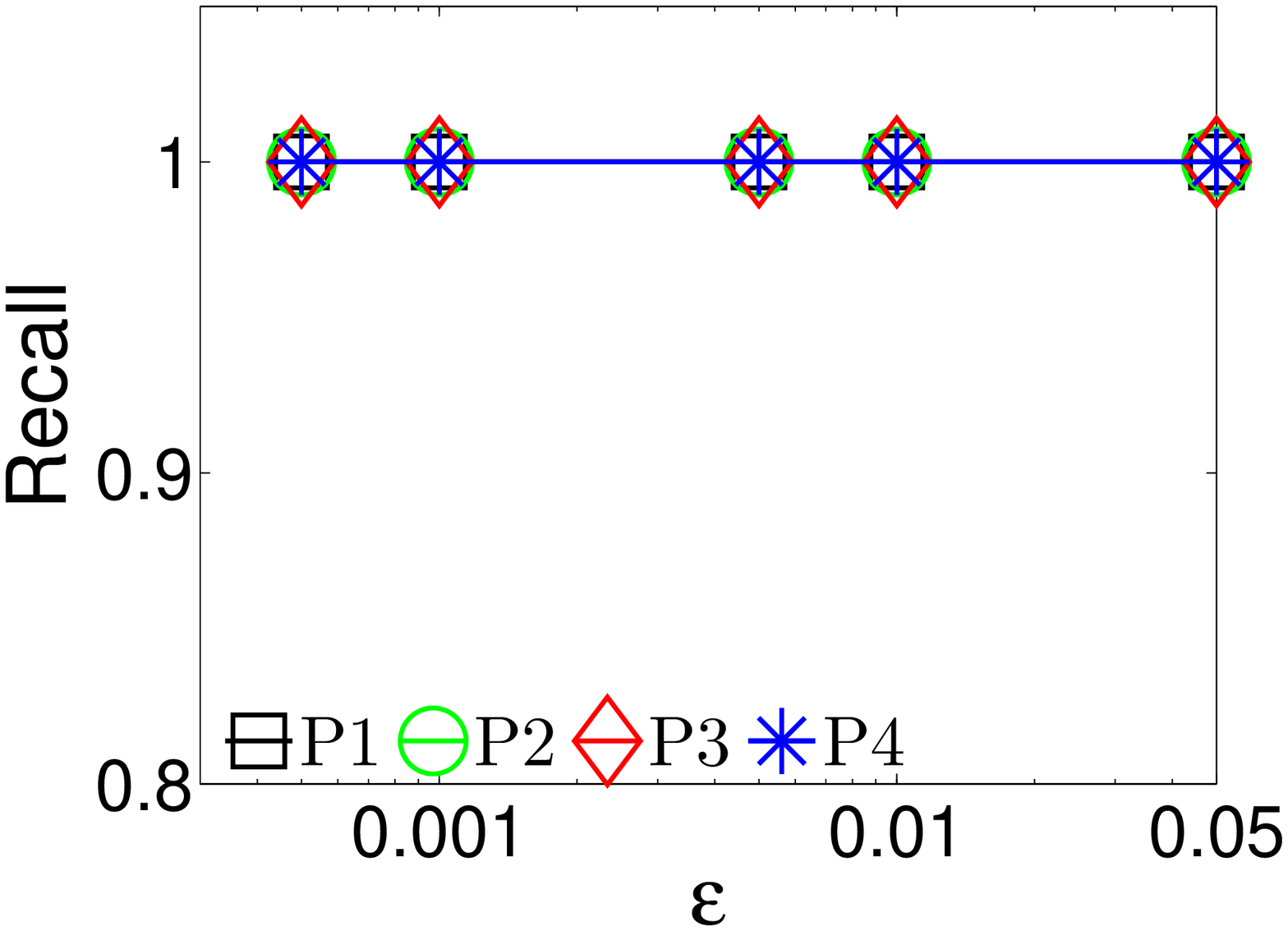}
\label{fig:zipf_recall_eps}
}\vspace{2mm}
\subfigure[precision vs. $\eps$.]{
\includegraphics[width=\figsize]{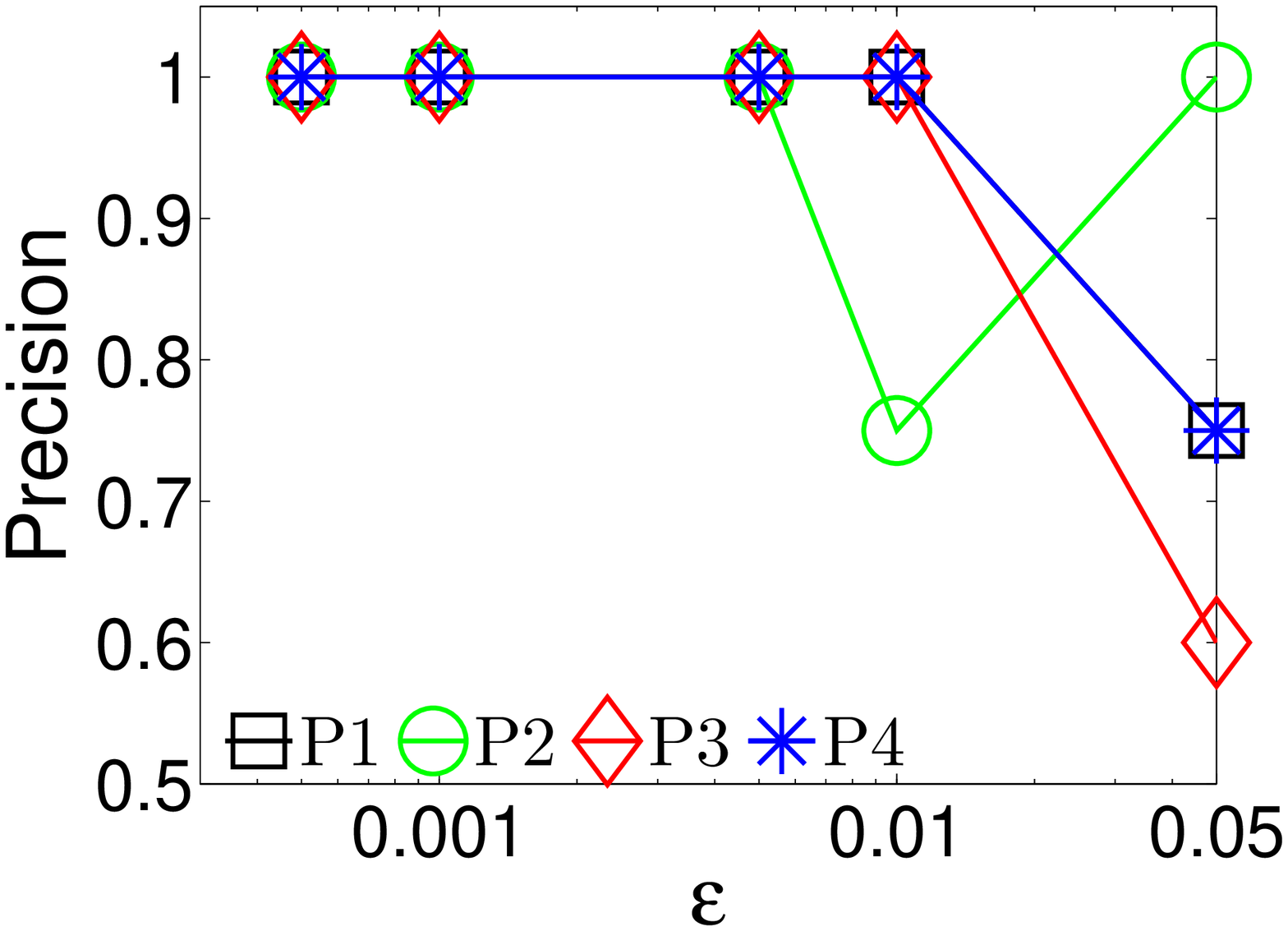}
\label{fig:zipf_prec_eps}
}\vspace{2mm} \subfigure[\s{err} vs. $\eps$.]{
\includegraphics[width=\figsize]{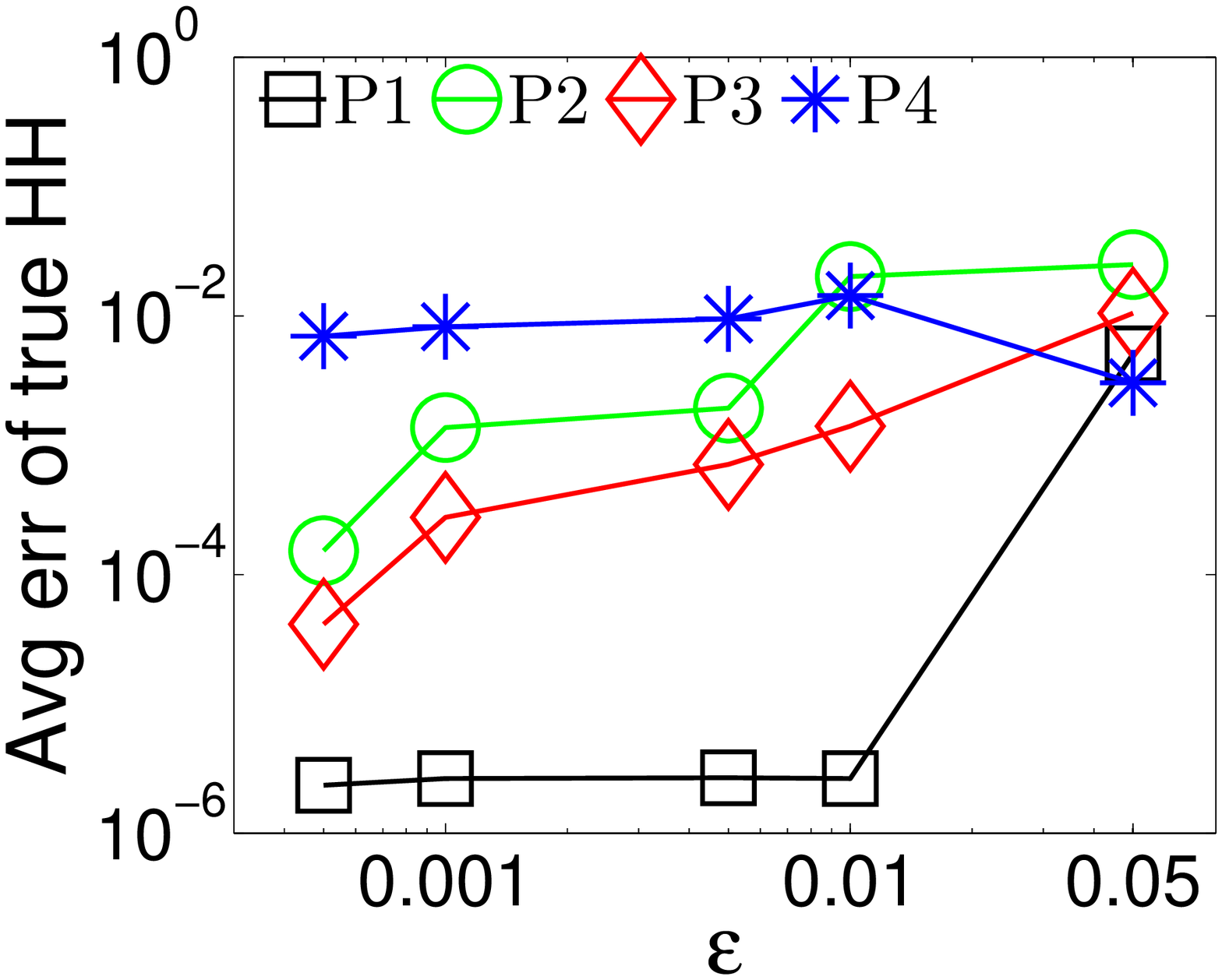}
\label{fig:zipf_errtrue_eps}
}\vspace{2mm}
\subfigure[\s{msg} vs. $\eps$.]{
\includegraphics[width=\figsize]{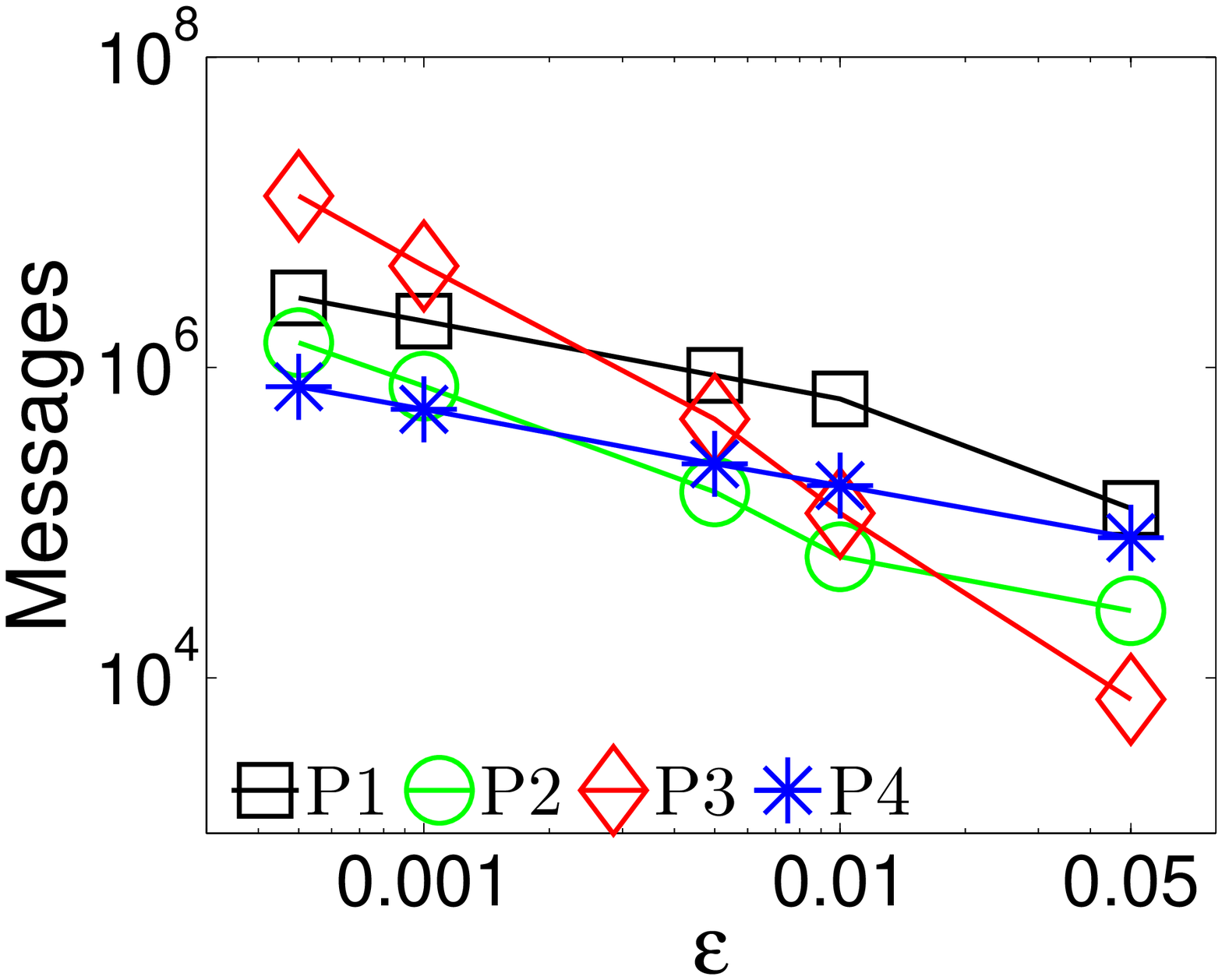}
\label{fig:zipf_msg_eps}
}\vspace{2mm} \subfigure[\s{err} vs. \s{msg}]{
\includegraphics[width=\figsize]{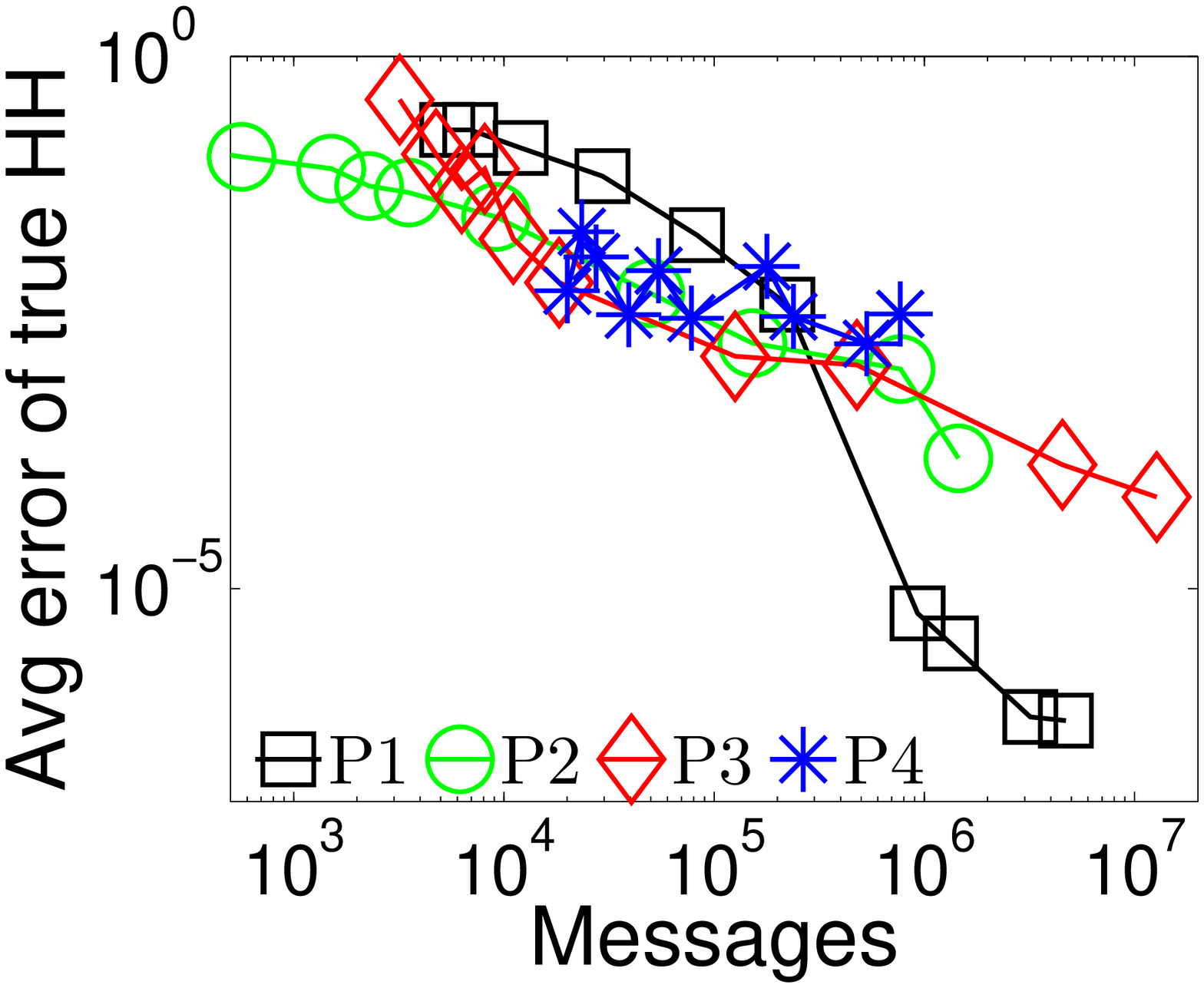}
\label{fig:zipf_msg_err}
}
\subfigure[\s{msg} vs. $\beta$.]{
\includegraphics[width=\figsize, height=1.1in]{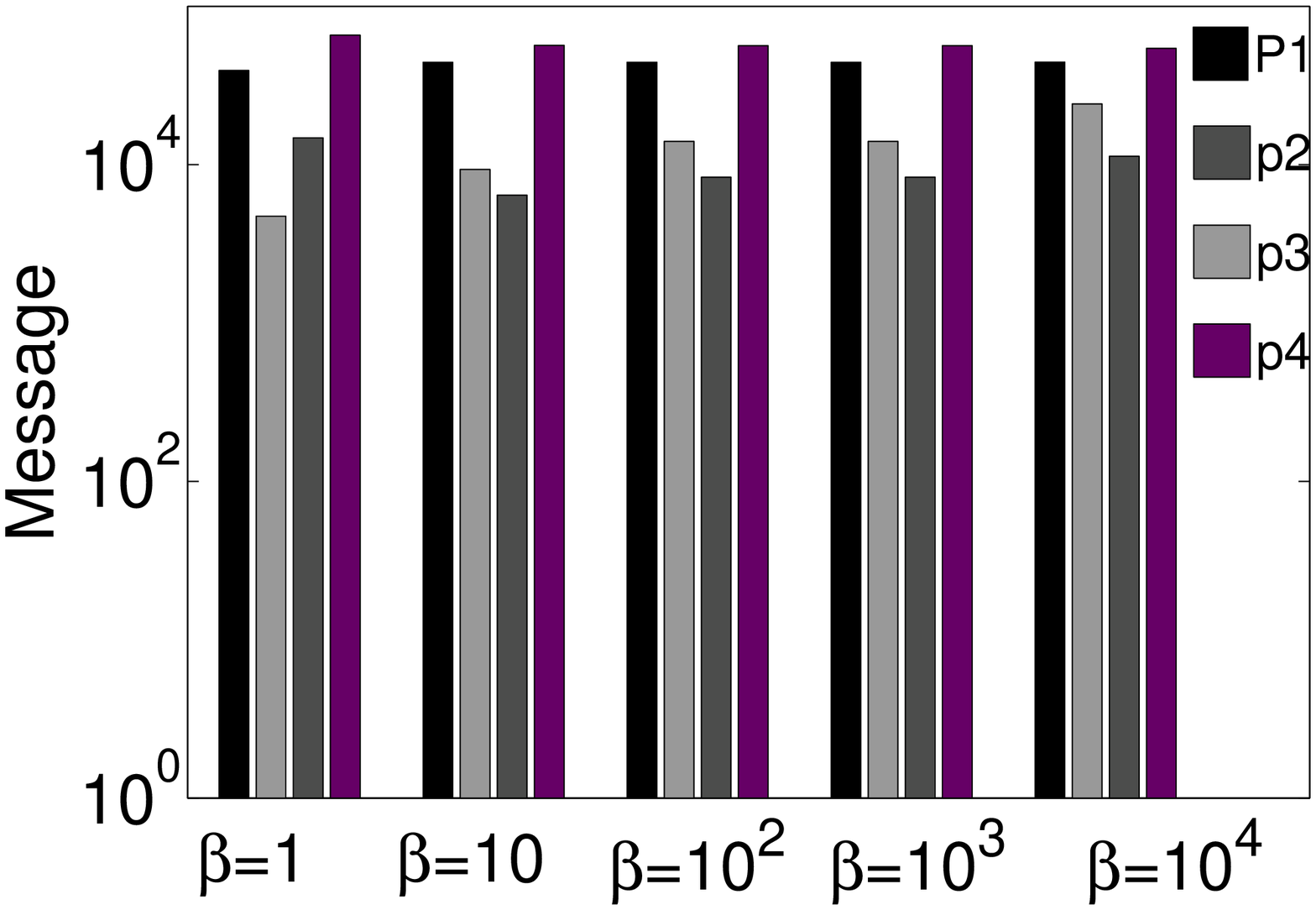}
\label{fig:zipf_beta_msg}
}
\caption{Results for distributed weighted heavy hitters protocols on Zipfian distribution with skew=2.}
\label{fig:zipfian}\vspace{-3mm}
\end{centering}
\end{figure}

\subsection{Distributed Weighted Heavy Hitters}
We denote four protocols for tracking distributed weighted heavy
hitters as \piw, \piiw, \piiiw\ and \piiiiw\ respectively.  As a
baseline, we could send all $10^7$ stream elements to the coordinator,
this would have no error.  All of our heavy hitters protocols return
an element $e$ as heavy hitter only if $\hat{W}_e/\hat{W} \geq \phi -
\eps/2$ while the exact weighted heavy hitter method which our
protocols are compared against, returns $e$ as heavy hitter if
$f_e(A)/W \geq \phi$.

We set the heavy-hitter threshold $\phi$ to $0.05$ and we varied error
guarantee $\eps$ in the range $\{5\times10^{-4}, 10^{-3}, 5
\times10^{-3},$ $10^{-2}, 5\times10^{-2}\}$. When the plots do not
vary $\eps$, we use the default value of $\eps = 10^{-3}$. Also we
varied number of sites ($m$) from $10$ to $100$, otherwise we have as
default $m = 50$.

All four algorithms prove to be highly effective in estimating
weighted heavy hitters accurately, as shown in \emph{recall} (Figure
\ref{fig:zipf_recall_eps}) and \emph{precision} (Figure
\ref{fig:zipf_prec_eps}) plots. In particular, the recall values for
all algorithms are constant $1.0$.

Note that precision values dip, but this is because the true heavy
hitters have $f_e(A)/W$ above $\phi$ where our algorithms only return
a value if $\hat W_e/\hat{W} \geq \phi - \eps/2$, so they return more
false positives as $\eps$ increases.  For $\eps$ smaller than $0.01$,
all protocols have a precision of $1.0$.

When measuring (the measured) \s{err} as seen in Figure
\ref{fig:zipf_errtrue_eps}, our protocols consistently outperform the
error parameter $\eps$.  The only exception is $\piiiiw$, which has
slightly larger error than predicted for very small $\eps$; recall
this algorithm is randomized and has a constant probability of failure.
\piw\ has almost no error for $\eps = 0.01$ and below; this can be
explained by improved analysis for Misra-Gries~\cite{berinde10:_space}
on skewed data, which applies to our Zipfian data.  Protocols \piiw\
and \piiiw\ also greatly underperform their guaranteed error.

The protocols are quite communication efficient, saving several orders
of magnitude in communication as shown in Figure
\ref{fig:zipf_msg_eps}.  For instance, all protocols use roughly
$10^5$ messages at $\eps=0.01$ out of $10^7$ total stream elements. To
further understand different protocols, we tried to compare them by
making them using (roughly) the same number of messages. This is
achieved by using {\em different} $\eps$ values. As shown in Figure
\ref{fig:zipf_msg_err}; all protocols achieved excellent approximation
quality, and the measured error drops quickly as we allocate more
budget for the number of messages.  In particular, \piiw\ is the best
if fewer than $10^5$ messages are acceptable with \piiiw\ also shown
to be quite effective.  \piw\ performs best if $10^6$ messages are
acceptable.

In another experiment, we tuned all protocols to obtain (roughly) the
same measured error of \s{err} = $0.1$ to compare their communication
cost versus the upper bound on the element weights ($\beta$). Figure
\ref{fig:zipf_beta_msg} shows that they are all robust to the
parameter $\beta$; \piiw\ or \piiiw\ performs the best.

\subsection{Distributed Matrix Tracking}
The strong results for distributed weighted heavy hitter protocols
transfer over empirically to the results for distributed matrix
tracking, but with slightly different trade-offs.  Again, we denote
our three protocols by \pim, \piim, and \piiim\ in all plots.  {\em As
  a baseline}, we consider two algorithms: they both send all data to
the coordinator. One calls Frequent-Directions (\textsf{FD})
~\cite{Lib12}, and second calls \textsf{SVD} which is optimal but not
streaming.  In all remaining experiments, we have used default value
$\eps= 0.1$ and $m = 50$, unless specified.  Otherwise $\eps$ varied
in range $\{5\times10^{-3}, 10^{-2}, 5 \times10^{-2}, 10^{-1},
5\times10^{-1}\}$, and $m$ varied in range $[10, 100]$.

\begin{table}[tbp]
\begin{center}
\begin{tabular}{|c||c|c||c|c|}
\hline
\textbf{DataSet} & \multicolumn{2}{|c|}{\textbf{PAMAP}, $k=30$} & \multicolumn{2}{|c|}{\textbf{MSD}, $k=50$} \\
\hline 
\textbf{Method} & \text{\s{err}} & \text{\s{msg}} &  \text{\s{err}} & \text{\s{msg}} \\ 
\hline 
\hline
\pim & 7.5859e-06 & 628537  & 0.0057 & 300195  \\ 
\hline
\piim & 0.0265 & 10178  & 0.0695 & 6362\\ 
\hline 
\piiim\textsc{wor} & 0.0057 &  3962  & 0.0189 & 3181 \\ 
\hline 
\piiim\textsc{wr} & 0.0323 &  25555  & 0.0255 & 22964 \\ 
\hline 
\textsf{FD} & 2.1207e-004 & 629250   & 0.0976 & 300000  \\ 
\hline
\textsf{SVD} & 1.9552e-006 & 629250   & 0.0057 & 300000 \\ 
\hline 
\end{tabular} 
\end{center}  \vspace{-5mm}
\caption{\label{tbl:pamap_centr}
Raw numbers of PAMAP and MSD.}\vspace{-4mm}
\end{table}

Table \ref{tbl:pamap_centr} compares all algorithms, including \s{SVD}
and \s{FD} to compute rank $k$ approximations of the matrices, with $k
= 30$ and $k=50$ on PAMAP and MSD respectively.  Since \s{err} values
for the two offline algorithms are minuscule for PAMAP, it indicates
it is a low rank matrix (less than 30), where as MSD is high rank,
since error remains, even with the best rank 50 approximation from the
\s{SVD} method.

Note that \piiim\textsc{wor} and \piiim\textsc{wr} refer to Protocol
3, {\em without replacement} and {\em with replacement} sampling
strategies, respectively.  As predicted by the theoretical analysis,
we see that \piiim\textsc{wor} outperforms \piiim\textsc{wr} in both
settings, always having much less error and many fewer messages.
Moreover, \piiim\textsc{wor} will gracefully shift to sending all data
deterministically with no error as $\eps$ becomes very
small.
Hence we only use \piiim\textsc{wor} elsewhere, labeled as just
\piiim.

Also note that \pim\ in the matrix scenario is far less effective;
although it achieves very small error, it sends as many messages (or
more) as the naive algorithms.  Little compression is taking place by
\s{FD} at distributed sites before the squared norm threshold is
reached.

\begin{figure}
\begin{centering}
\subfigure[\s{err} vs. $\eps$.]{
\includegraphics[width=\figsize]{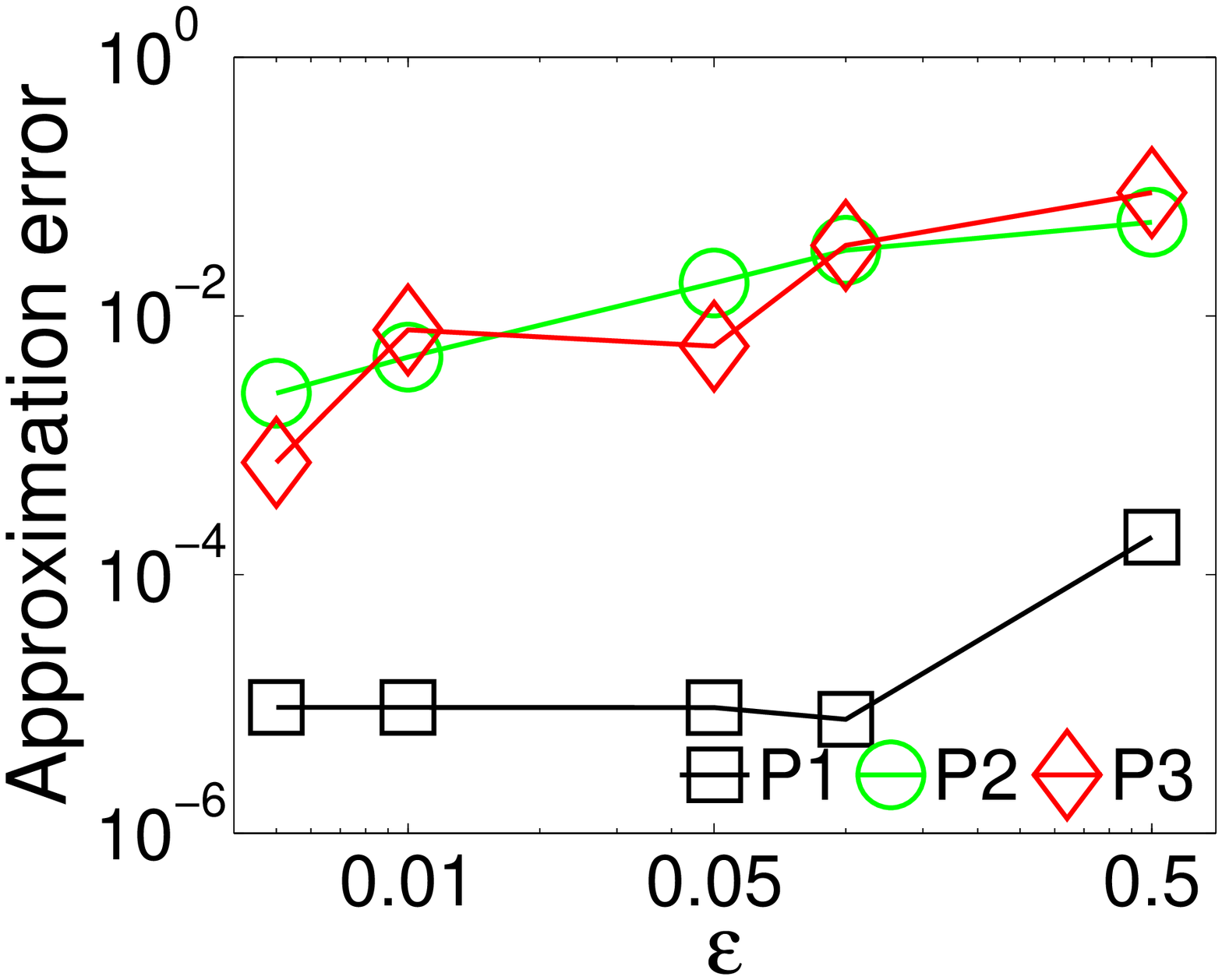}
\label{fig:pamap_err_eps}
}\vspace{2mm}
\subfigure[\s{msg} vs. $\eps$.]{
\includegraphics[width=\figsize, height=1.1in]{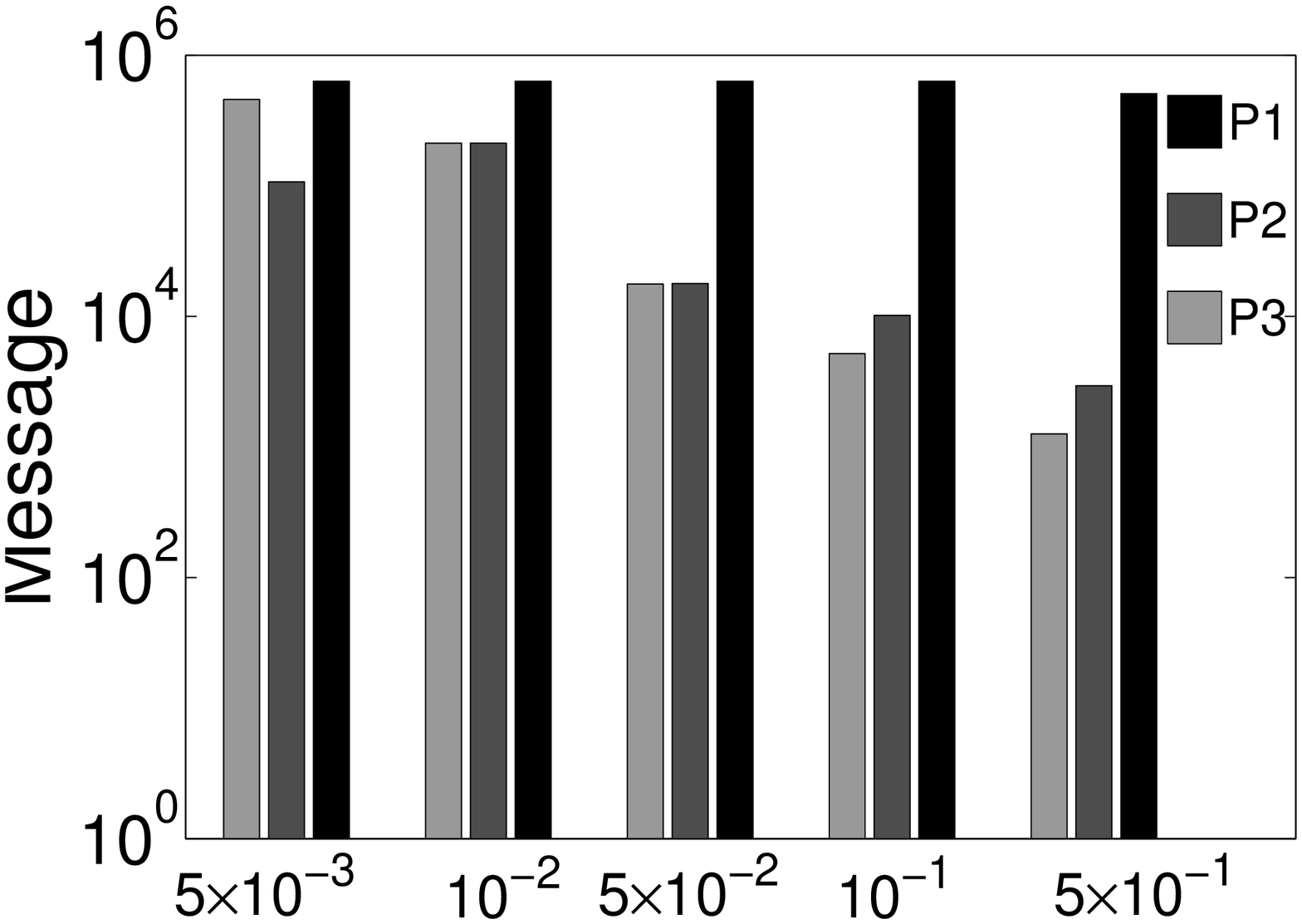}
\label{fig:pamap_msg_eps}
}\vspace{2mm}
\subfigure[\s{msg} vs. site]{
\includegraphics[width=\figsize]{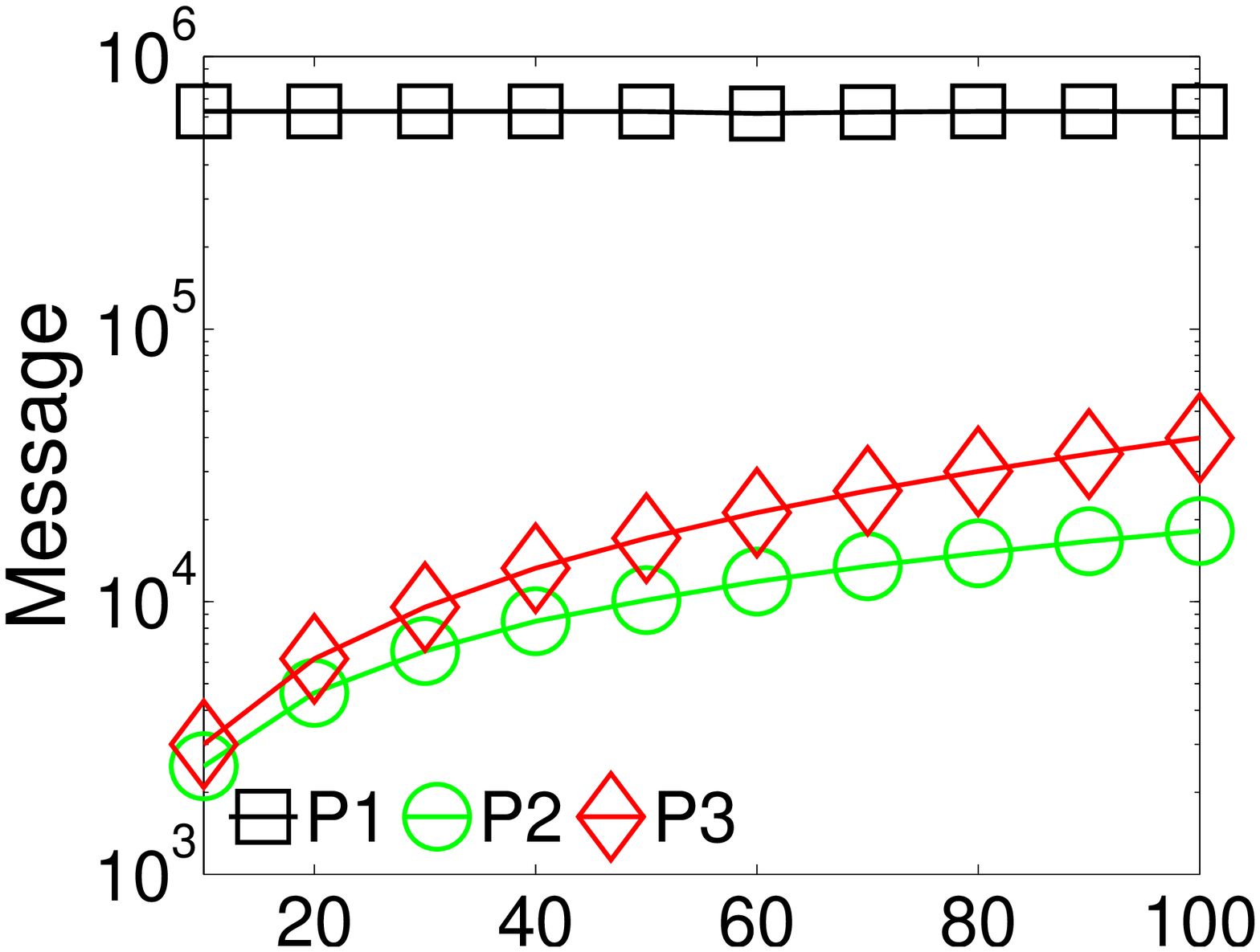}
\label{fig:pamap_msg_sites}
}
\subfigure[\s{err} vs. site]{
\includegraphics[width=\figsize]{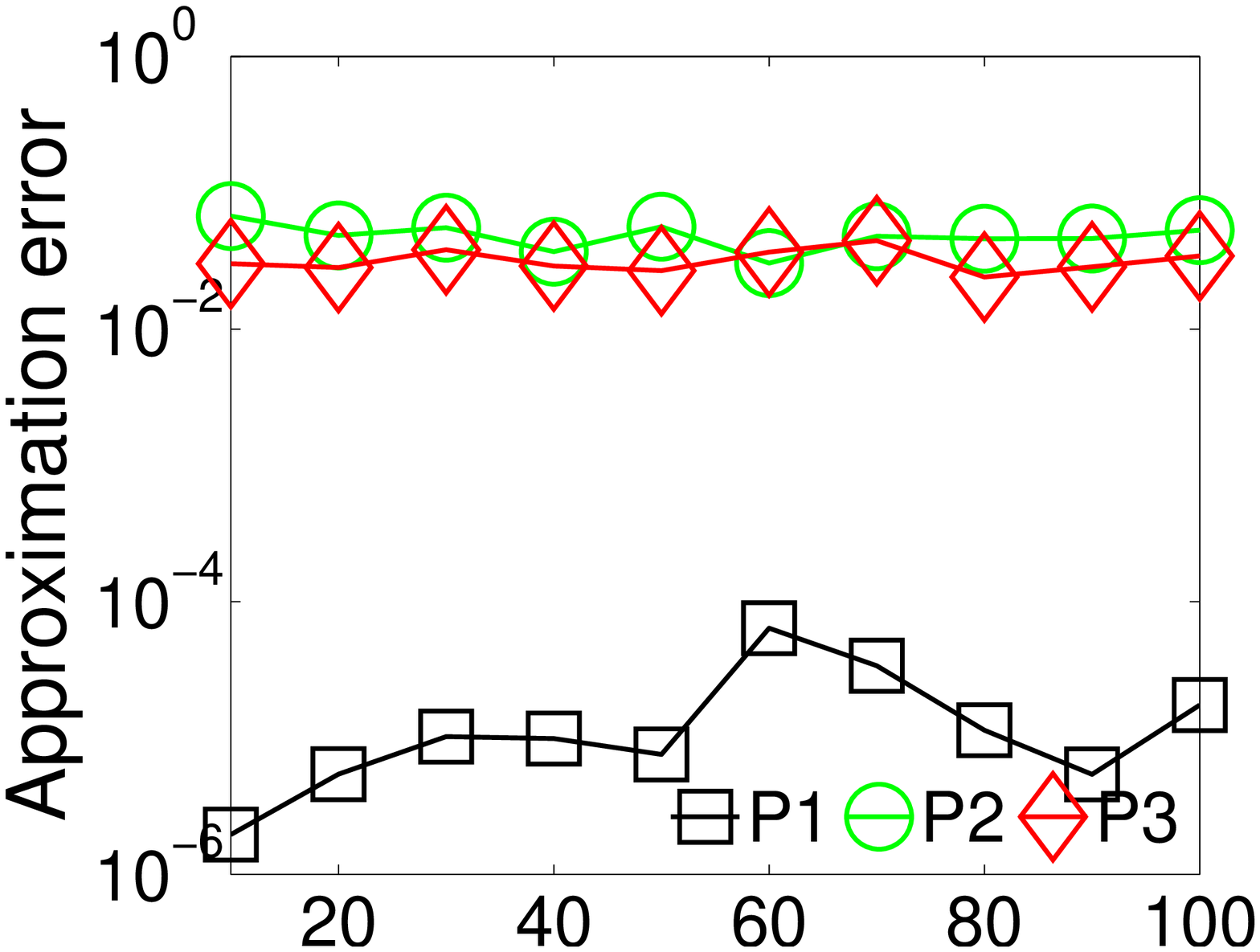}
\label{fig:pamap_err_sites}
}
\caption{Experiments for PAMAP dataset}
\label{fig:pamap}\vspace{-3mm}
\end{centering}
\end{figure}

Figures \ref{fig:pamap_err_eps} and \ref{fig:yp_err_eps} show as
$\eps$ increases, error of protocols increases too. In case of \piiim~
this observation is justified by the fact \piiim\ samples
$O((1/\eps^2) \log(1/\eps))$ elements, and as $\eps$ increases, it
samples fewer elements, hence results in a weaker estimation of true
heavy directions. In case of \piim, as $\eps$ increases, they allocate
a larger error slack to each site and sites communicate less with the
coordinator, leading to a coarse estimation.  Note that again \pim\
vastly outperforms its error guarantees, this time likely explained
via the improved analysis of Frequent-Directions~\cite{GP14}.

Figures \ref{fig:pamap_msg_eps} and \ref{fig:yp_msg_eps} show number
of messages of each protocol vs. error guarantee $\eps$. As we see, in
large values of $\eps$ (say for $\eps > 1/m = 0.02$), \piim\ typically
uses slightly more messages than
\piiim. 
But as $\eps$ decreases, \piiim\ surpasses \piim\ in number of
messages. This confirms the dependency of their asymptotic bound on
$\eps$ ($1/\eps^2$ vs. $1/\eps$). \pim\ generally sends much more
messages than both \piim\ and \piiim.

Next, we examined the number of sites ($m$). Figures
\ref{fig:pamap_msg_sites} and \ref{fig:yp_msg_sites} show that \piim\
and \piiim\ used more communication as $m$ increases, showing a linear
trend with respect to $m$.  \pim\ shows no trend since its
communication depends solely on the total weight of the stream. Note
that \pim\ sends its whole sketch, hence fix number of messages,
whenever it reaches the threshold.  As expected, the number of sites
does not have significant impact on the measured approximation error
in any protocol; see Figures \ref{fig:pamap_err_sites} and
\ref{fig:yp_err_sites}.

\begin{figure}
\begin{centering}
\subfigure[\s{err} vs. $\eps$.]{
\includegraphics[width=\figsize]{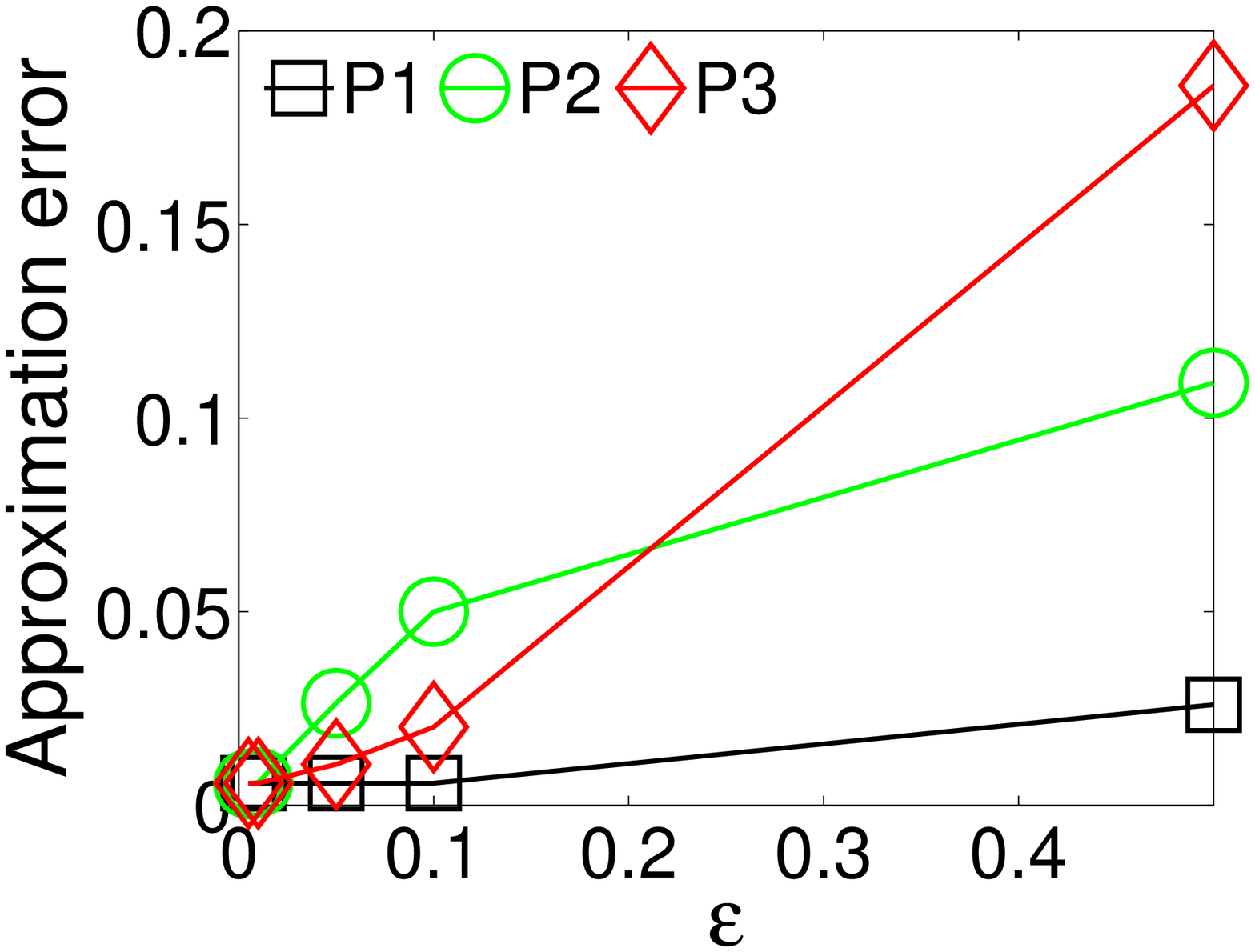}
\label{fig:yp_err_eps}
}\vspace{2mm}
\subfigure[\s{msg} vs. $\eps$.]{
\includegraphics[width=\figsize]{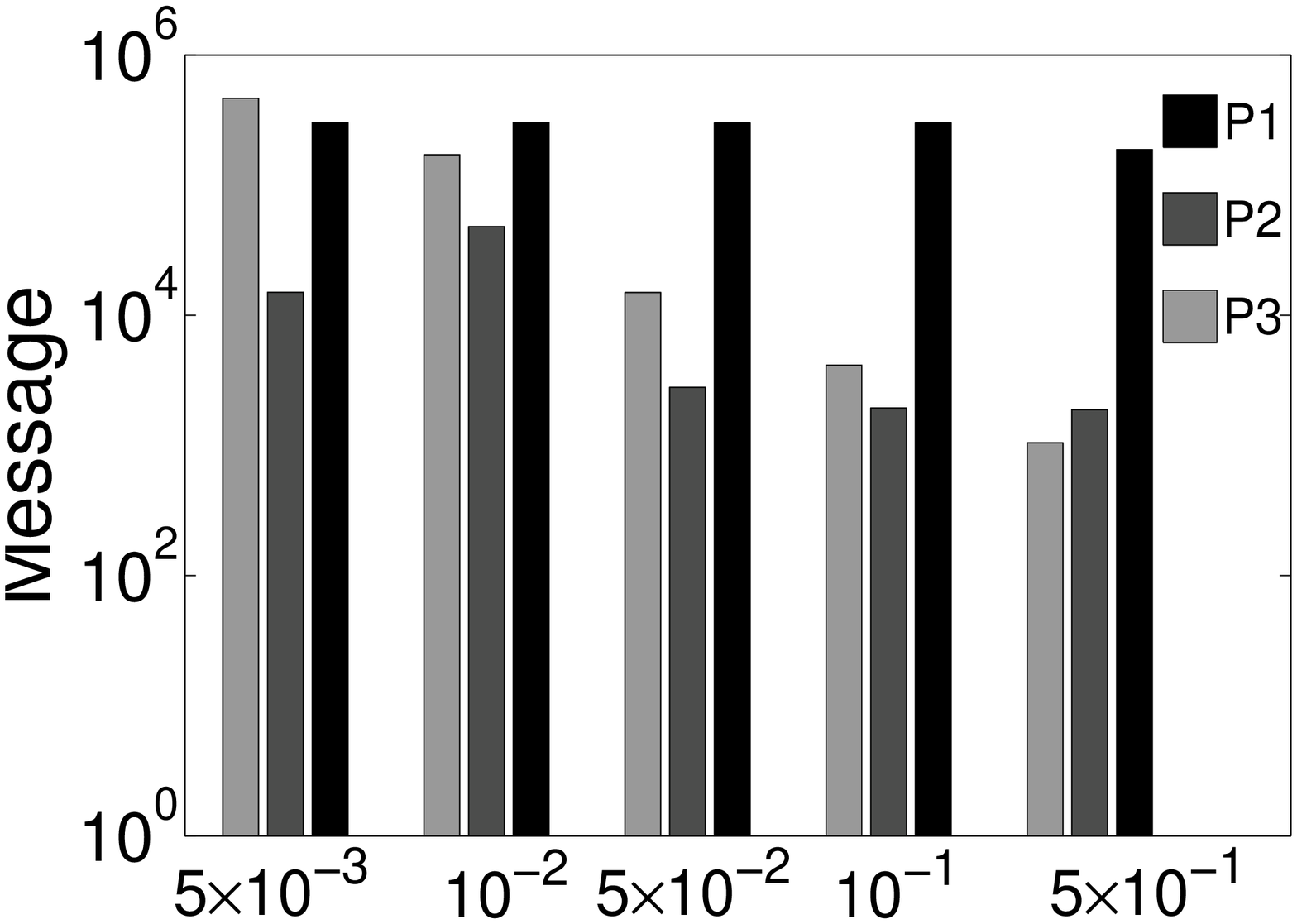}
\label{fig:yp_msg_eps}
}\vspace{2mm}
\subfigure[\s{msg} vs. site]{
\includegraphics[width=\figsize]{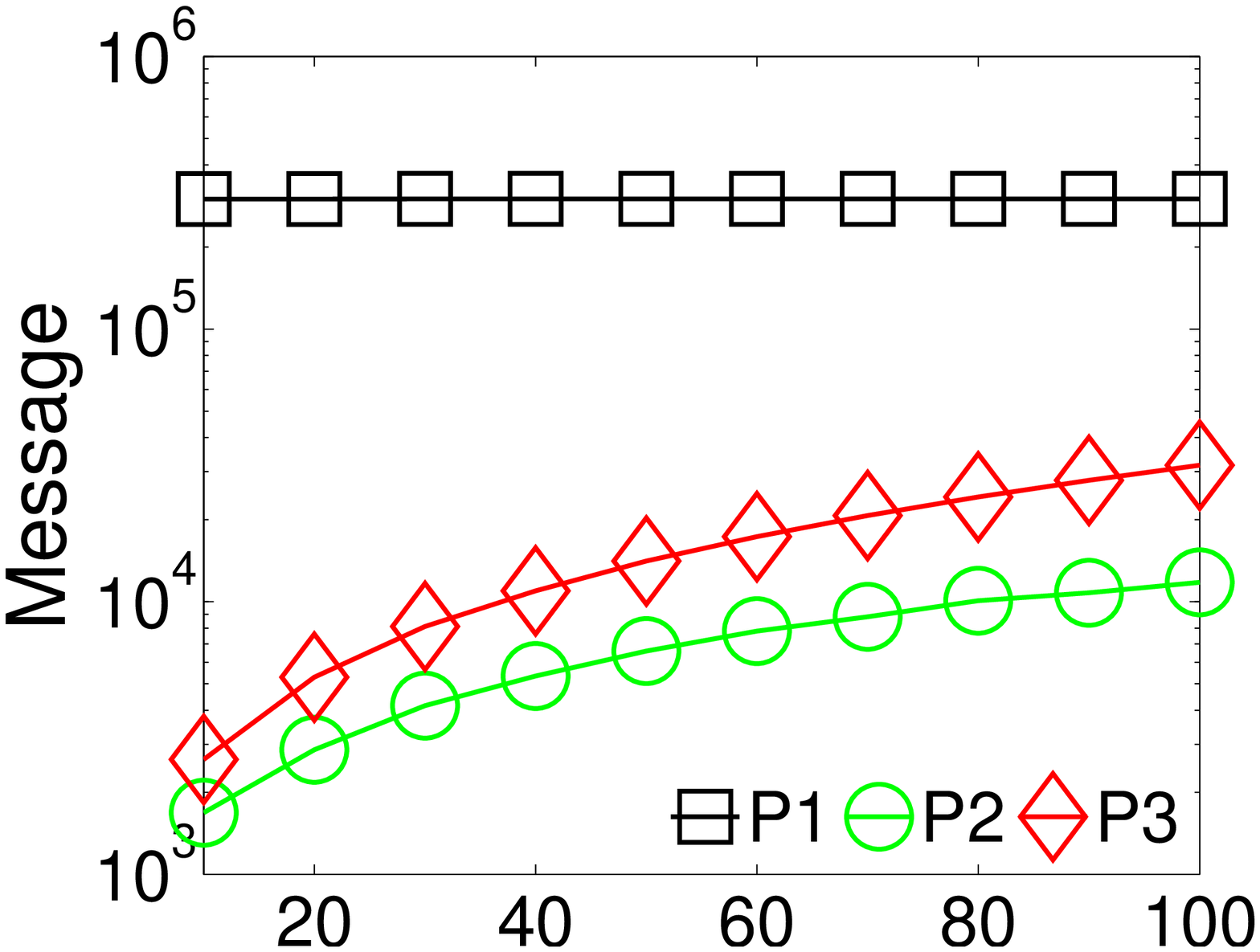}
\label{fig:yp_msg_sites}
}
\subfigure[\s{err} vs. site]{
\includegraphics[width=\figsize]{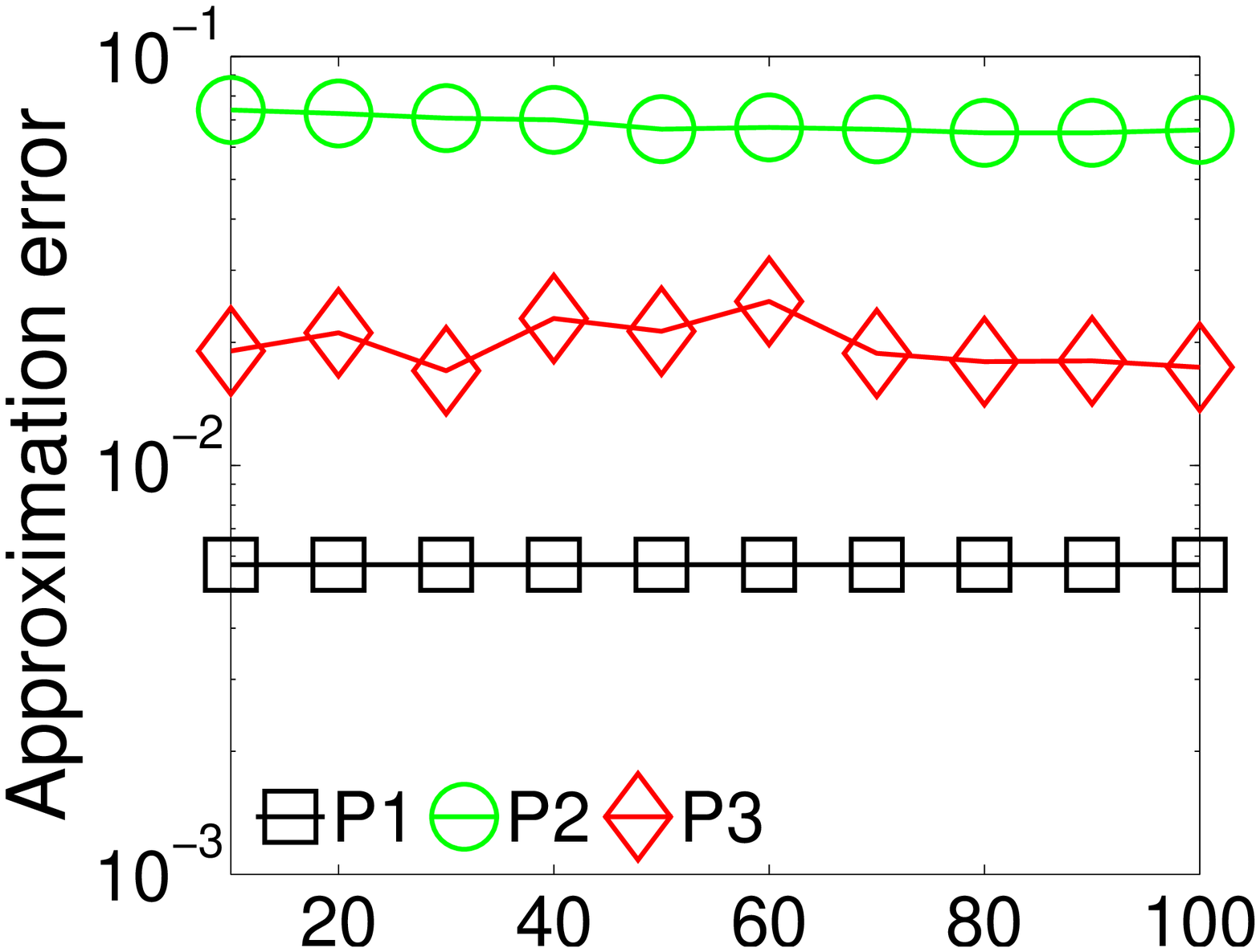}
\label{fig:yp_err_sites}
}
\caption{Experiments for MSD dataset}\vspace{-3mm}
\label{fig:yp}
\end{centering}
\end{figure}

We also compared the performance of protocols by tuning the $\eps$
parameter to achieve (roughly) the same measured error. Figure
\ref{fig:compare} shows their communication cost (\#\s{msg}) vs the
\s{err}.  As shown, protocols \pim, \piim, and \piiim\ incur less
error with more communication and each works better in various regimes
of the \s{err} versus \s{msg} trade-off.  \pim\ works the best when
the smallest error is required, but more communication is permitted.
Even though its communication is the same as the naive algorithms in
these examples, it allows each site and the coordinator to run small
space algorithms.  For smaller communication requirements (several of
orders of magnitude smaller than the naive methods), then either
\piim\ or \piiim\ are recommended.  \piim\ is deterministic, but
\piiim\ is slightly easier to implement. Note that since MSD is high
rank, and even the naive \s{SVD} or \s{FD} do not achieve really
small error (e.g. $10^{-3}$), it is not surprising that our algorithms
do not either.

\begin{figure}
\begin{centering}
\subfigure[PAMAP]{
\includegraphics[width=\figsize]{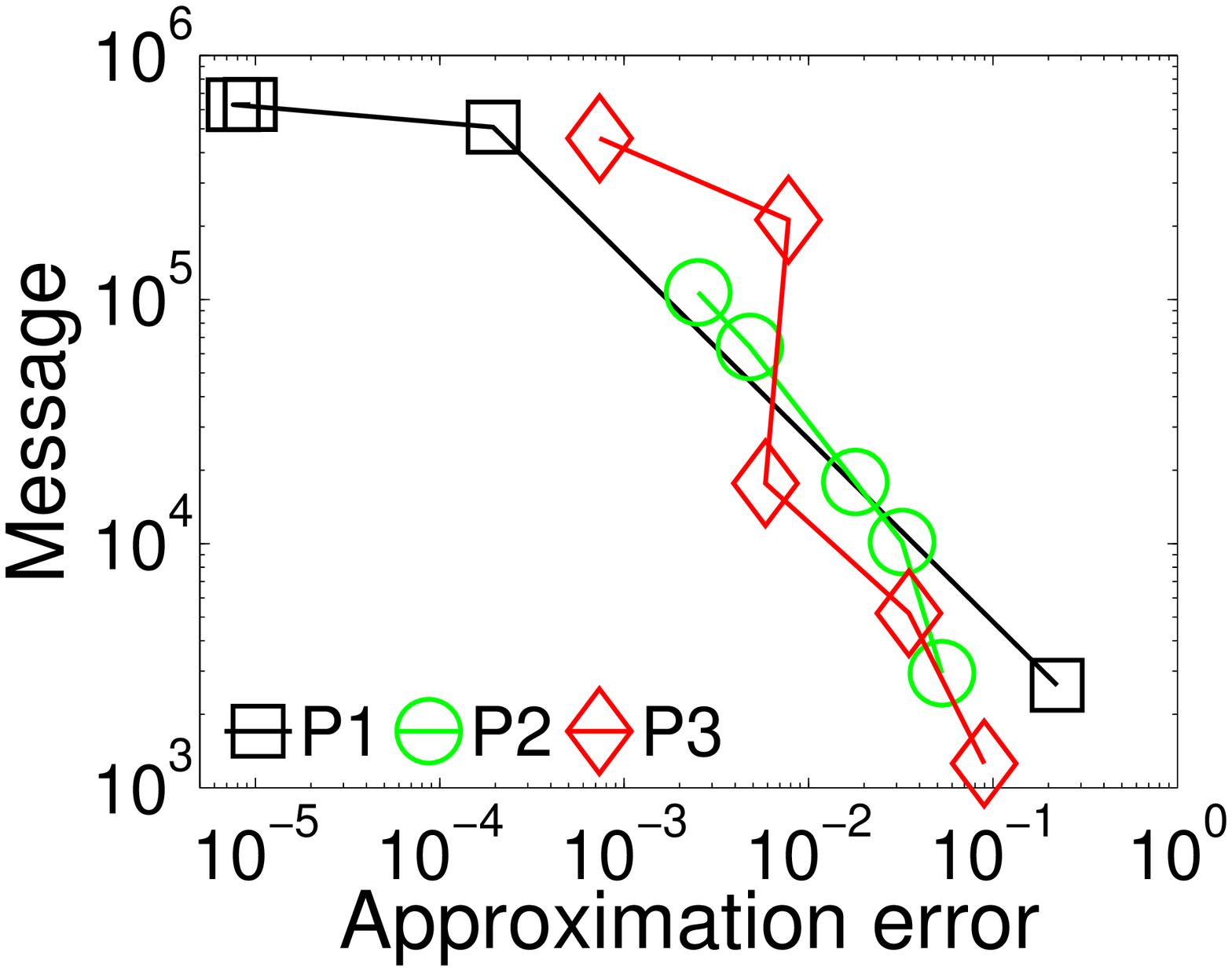}
\label{fig:pamap_msg_err}
}
\subfigure[MSD]{
\includegraphics[width=\figsize]{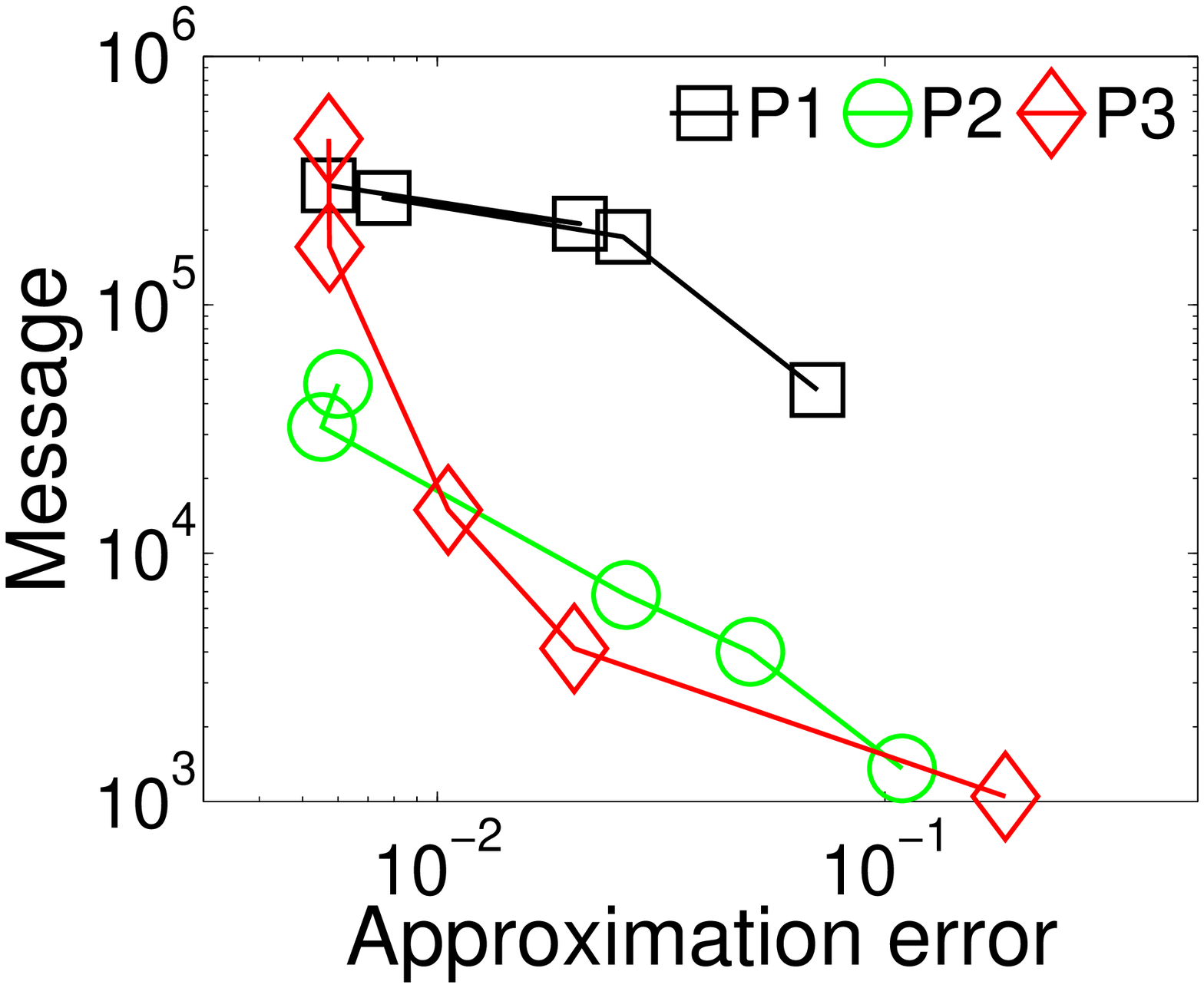}
\label{fig:yp_msg_err}
}
\vspace{-1mm}
\caption{Comparing the two protocols: \s{msg} vs. \s{err}}\vspace{-4mm}
\label{fig:compare}
\end{centering}
\end{figure}

\vspace{-1mm}
\section{Conclusion}
\label{sec:conclude} 
We provide the first protocols for monitoring weighted heavy
hitters and matrices in a distributed stream.  They are backed by
theoretical bounds and large-scale experiments. Our results are based
on important connections we establish between the two problems.  
Interesting open problems include, but are not
limited to, extending our results to the sliding window model, and
investigating distributed matrices that are column-wise distributed
(i.e., each site reports values from a fixed column in a matrix).

\vspace{-2mm}

\small 
 \bibliographystyle{abbrv}
\bibliography{mina,discrepancy,HH-mina}

\normalsize

\begin{appendix}
\label{sec:appendix}
\section{Proof of Lemma 4}
\begin{proof}

  In order to reach round $j$ we must have $s$ items with priority
  $\rho_n > \tau_j = 2^j$; this happens with probability $\min(1, w_n
  / 2^j) \leq \min(1, \beta /2^j) \leq \beta /2^j$.  We assume for now
  $\beta \leq 2^j$ for values of $j$ we consider; the other case is
  addressed at the end of the proof.

  Let $X_{n,j}$ be a random variable that is $1$ if $\rho_n \geq
  \tau_j$ and $0$ otherwise.  Let $M_j = \sum_{n=1}^N X_{n,j}$.  Thus
  the expected number of items that have a priority greater than
  $\tau_j$ is $\E[M_j] = \sum_{n=1}^N w_n /2^j \leq N \beta /2^j$.
  Setting $j_N = \lceil \log_2(\beta N /s) \rceil$ then $\E[M_{j_N}]
  \leq s$.

  We now want to use the following Chernoff bound on the $N$
  independent random variables $X_{n,j}$ with $M = \sum_{n=1}^N
  X_{n,j}$ that bounds $ \Pr[M_j \geq (1+\alpha) \E[M_j] ] \leq \exp(-
  \alpha^2 \E[M_j] / (2+\alpha)).  $ 

Note that $\E[M_{j_N+1}] \leq (N
  \beta)/ 2^{j_N+1} \leq s /2$.  Then setting $\alpha = 1$ ensures that
\[
(1+\alpha) \E[M_{j_N+1}] = 2 \cdot (s/2) \leq s.
\]

Thus we can solve 

\begin{align*}
\Pr[M_{j_N+1} \geq s] 
&\leq 
\exp\left( - \frac{s/2}{3}\right) = \exp\left( - \frac{-s}{6}\right).
\end{align*}

Thus since in order to reach round $j_N+1 = \log_2(\beta N/s)+1 =
O(\log(\beta N/s))$, we need $s$ items to have priority greater than
$\tau_{j_N+1}$, this happens with probability at most $e^{-\Omega(s)}$.
Recall that to be able to ignore the case where $w_n > \tau_{j_N+1}$
we assumed that $\beta < \tau_{j_N+1}$.  If this were not true, then
$\beta > 2^{\log_2(\beta N/s)+1} = 2\beta N /s$ implies that
$s > N$, in which case we would send all elements before the end of the
first round, and the number of rounds is $1$.
\end{proof}

\section{Proof of Lemma 6}
\begin{proof}

  To prove our claim, we use the following Chernoff-Hoeffding
  bound~\cite{PS97}.  Given a set of negatively-correlated random
  variables $Y_1, \ldots, Y_r$ and $Y = \sum_{n=1}^r Y_n$, where each
  $Y_n \in [a_i,b_i]$, where 
$\Delta = \max_n (b_n -a_n)$ then $\Pr[ |Y - \E[Y]| \geq \alpha ] \leq \exp(-2\alpha^2/ r \Delta^2)$.

  For a given item $e \in [u]$ and any pair $(a_n,w_n) \in S$, we
  define a random variable%
  \footnote{Note that the random variables $X_{n,e}$ are
    negatively-correlated, and not independent, since they are derived
    from a sample set $S$ drawn without replacement.  Thus we appeal to
    the special extension of the Chernoff-Hoeffding bound by Panconesi
    and Srinivasan~\cite{PS97} that handles this case.  We could of
    course use with-replacement sampling, but these algorithms require
    more communication and typically provide worse bounds in practice, as demonstrated.
  } $X_{n,e}$ as follows:
\begin{equation*}
X_{n,e}= \bar{w}_n\textrm{ if } a_n=e, \textrm{ }0 \textrm{ otherwise}.
\end{equation*}
Define a heavy set $H = \{a_n \in A \mid w_n \geq \tau_j\}$, these
items are included in $S$ deterministically in round $j$.  Let the
light set be defined $L = A \setminus H$.  Note that for each $a_n \in
H$ that $X_{n,e}$ is deterministic, given $e$.  For all $a_n \in S
\cap L$, then $X_{n,e} \in [0,2\tau_j]$ and hence using these as random
variables in the Chernoff-Hoeffding bound, we can set $\Delta =
2\tau_j$.

Define $M_{e} = \sum_{a_n \in S} X_{n,e}$, and note that $f_e(S) =
M_e$ is the estimate from $S'$ of $f_e(A)$.  Let $W_L = \sum_{a_n \in
  L} w_i$.  Since all light elements are chosen with probability
proportionally to their weight, then given an $X_{n,e}$ for $a_i \in S
\cap L$ it has label $e$ with probability $f_e(L) / W_L$.  And in
general $\E[\sum_{a_n \in S \cap L} \bar{w}_n] = \E[W_{S \cap L}] =
W_L$.
Let $H_e = \{a_n \in H \mid a_n = e\}$.  Now we can see
\begin{align*}
\E[f_e(S)] = \E\left[ \sum_{n=1}^{|S|} X_{n,e} \right] 
&= 
\sum_{a_n \in H_e} w_n + \E\left[\sum_{a_n \in S \cap L} \tau_j\right] \cdot \frac{f_e(L)}{W_L}
\\ & = 
f_e(H) + W_L \cdot \frac{f_e(L)}{W_L} 
= 
f_e(A).
\end{align*}

Now we can apply the Chernoff-Hoeffding bound.  
\begin{align*}
&\Pr \left(\left|f_e(S) - f_e(A)\right | > \eps W_A \right) 
= 
\Pr \left(\left |M_e-E[M_e] \right | > \eps W_A \right) 
\\ &\leq 
\exp \left(\frac{-2\eps^2 W_A^2}{|S| 4 \tau_j^2}\right)  
\leq
\exp \left(-\frac{1}{2} \eps^2 |S|/(1+\eps)^2 \right)  
\leq \delta, 
\end{align*}
where the last line follows since $(1+\eps) W_A \geq \sum_{a_n \in S}
\bar w_n \geq |S| \tau_j$, where the first inequality holds with high
probability on $|S|$.  Solving for $|S|$ yields $|S| \geq \frac{(1+\eps)^2}{2\eps^2} \ln (1/\delta)$.  Setting $\delta = O(\eps^2 /\log (1/\eps)) = 1/|S|$ allows the result to hold with probability at least $1-1/|S|$.  
\end{proof}

\section{Distributed Matrix Tracking Protocol 4} 
\label{sec:P4m}

Again treating each row $a_i$ as having weight $w_i = \|a_i\|^2$, then
to mimic the weighted heavy-hitters protocol 4 we want to select each
row with probability $\hat p = 1-e^{-p \|a_i\|^2}$.  Here $p =
2\sqrt{m}/ (\eps \hat F)$ represents the probability to send a weight
$1$ item and $\hat F$ is a $2$-approximation of $\|A\|_F^2$ (i.e.
$\hat F \leq \|A\|_F^2 \leq 2 \hat F$) and is maintained and provided
by the coordinator.  We then only want to send a message from the
coordinator if that row is selected, and then it follows from the
analysis in Section \ref{sec:P3w} that in total $O((\sqrt{m}/\eps)
\log (\beta N))$ messages are sent, since $O(\sqrt{m}/\eps)$ messages
are sent each round in between $\hat F$ doubling and being distributed
by the coordinator, and there are $O(\log (\beta N))$ such rounds.

But replicating the approximation guarantees of protocol 4 is hard.
In Algorithm \ref{alg:P3w-site}, on each message a particular element
$e$ has its count updated \emph{exactly} with respect to a site $j$.
Because of this, we only need to bound the expected weight of stream
elements until another exact update is seen (at $1/p$) and then to
compensate for this we increase this weight by $1/p$ so it has the
right expected value.  It also follows that the variances are bounded
by $1/p^2$, and thus when $p$ is set $\Theta(\sqrt{m}/(\eps W))$ we
get at most $\eps W$ error.

Thus the most critical part is to update the representation (of local
matrices from $m$ sites) on the coordinator so it is exact for some
query.  We show that this can only be done for a limited set of
queries (along certain singular vectors), provide an algorithm to do
so, and then show that this is not sufficient for any approximation
guarantees.

\Paragraph{Replicated algorithm for matrices.}
Each site can keep track of $A_j$ the exact matrix describing all of
its data, and an approximate matrix $\hat A_j$.  The matrix $\hat A_j$
will also be kept on the coordinator for each site.  So the
coordinator's full approximation $\hat A = [\hat A_1; \hat A_2; \ldots
; \hat A_m]$ is just the stacking of the approximation from each site.
Since the coordinator can keep track of the contribution from each
site separately, the sites can maintain $\hat A_j$ under the same
process as the coordinator.
 
In more detail, both the site and the coordinator can maintain the
$[U,\Sigma,V] = \svd(\hat A_j)$, where $V = [v_1, v_2, \ldots, v_d]$
stores the right singular vectors and $\Sigma = \diag(s_1, s_2,
\ldots, s_d)$ are the singular values.  (Recall, $U$ is just an
orthogonal rotation, and does not change the squared norm.)  Thus
$\|\hat A_j x\|^2 = \sum_{i=1}^d s_i^2 \langle v_i, x\rangle^2$.  Now
if we can consider setting $A' = [\hat A_j; r]$ where $\|r\| = \langle v_{i'}, r\rangle$, so it is along the direction of a singular vector
$v_{i'}$, then
\vspace{-2mm}
\[
\|A' x\|^2 
= 
\|r\|^2 \langle v_{i'},x\rangle^2 + \sum_{i=1}^d s_i^2 \langle v_i, x \rangle^2.
\]
Thus if we update the singular value $s_{i'}$ to $\bar s_{i'} =
\sqrt{s_{i'}^2 + \|r\|^2}$, (and to simplify notation $\bar s_i = s_i$
for $i \neq i'$) then $\|A' x\|^2 = \sum_{i=1}^d (s'_i)^2 \langle v_i,
x \rangle^2$.  Hence, we can update the squared norm of $\hat A_j$ in
a particular direction, as long as that direction is one of its right
singular values.  But unfortunately, in general, for arbitrary
direction $x$ (if not along a right singular vector), we cannot do
this update while also preserving or controlling the orthogonal
components.

We can now explain how to use this form of update in a full protocol
on both the site and the coordinator.  The algorithm for the site is
outlined in Algorithm \ref{alg:P3m-site}.  On an incoming row $a$, we
updated $A_j = [A_j; a]$ and send a message with probability $1 -
e^{-p \|a\|^2}$ where $p = 2 \sqrt{m}/(\eps \hat F)$.  If we are
sending a
message, 
we first set $z_i = \sqrt{\|A_j v_i\|^2 + 1/p}$ for all $i \in [d]$,
and send a vector $z = (z_1, z_2, \ldots, z_d)$ to the coordinator.
We next produce the new $\hat A_j$ on both site and coordinator as
follows.  Set $Z = \diag(z_1, z_2, \ldots, z_d)$ and update $\hat A_j
= Z V^T$.  Now along any right singular vector $v_i$ of $\hat A_j$ we
have $\|\hat A_j v_i\|^2 = \|A_j v_i\|^2 + 1/p$.  Importantly note that
the right singular vectors of $\hat A_j$ do not change; although their
singular values and hence ordering may change, the basis does not.

\begin{algorithm}
\caption{\label{alg:P3m-site} P4: Site $j$ process new row $a$}
\begin{algorithmic}
\STATE Given $\hat F$ from coordinator, set $p = 2\sqrt{m}/ (\eps \hat F)$.  
\STATE The site also has maintained $[U,\Sigma,V] = \svd(\hat A_j)$.  
\STATE Update $A_j = [A_j; a]$.  
\STATE Set $\hat p = 1 - e^{-p \|a\|^2}$.  Generate $u \in \unif[0,1]$.  
\IF {($u \leq \hat p$)} 
  \STATE \textbf{for} $i \in [d]$ \textbf{do} $z_i = \sqrt{\|A_j v_i\|^2 + 1/p}$.  
  \STATE Send vector $z = (z_1, \ldots, z_d)$.  
  \STATE Set $Z = \diag(z_1, \ldots, z_d)$; update $\hat A_j = Z V^T$.
\ENDIF
\end{algorithmic}
\end{algorithm}

\begin{figure}[t!]
\begin{center}\includegraphics[width=.55\linewidth]{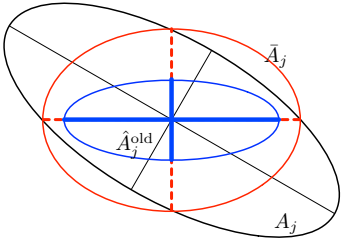}\end{center}
\vspace{-6mm} 
\caption{\label{fig:ellipse} {\small Possible update of $\hat
  A_j^{\text{old}}$ to $\bar A_j$ with respect to $A_j$ for $d=2$.
  Note the norm of a matrix along each direction $x$ is an ellipse,
  with the axes of the ellipse corresponding to the right singular
  vectors.  Thus $\hat A_j^{\text{old}}$ and $\bar A_j$ have the same
  axes, and along those axes both $\bar A_j$ and $A_j$ have the same
  norm, but otherwise are incomparable}.}\vspace{-3mm}
\end{figure}
\Paragraph{Error analysis.}
To understand the error analysis, we first consider a similar
protocol, except where instead of each $z_i$, we set $\bar z_i = \|A_j
v_i\|$ (without the $1/p$). Let $\bar Z = \diag(\bar z_1, \ldots, \bar
z_d)$ and $\bar A_j = \bar Z V^T$.  Now for all right singular vectors
$\|A_j v_i\|^2 = \|\bar A_j v_i\|^2$ (this is not true for general $x$
in place of $v_i$), and since $\bar z_i \geq s_i$ for all $i \in [d]$,
then for all $x$ we have $\|\bar A_j x \|^2 \geq \|\hat A_j^{\text{old}} x\|^2$, where $\hat A_j^{\text{old}}$ is the
approximation before the update.  See Figure \ref{fig:ellipse} to
illustrate these properties.

For directions $x$ that are right singular values of $\hat A_j$, this
analysis should work, since $\|\bar A_j x\|^2 = \|A_j x\|^2$.  But two
problems exist in other directions.  First in any other direction $x$
the norms $\|A_j x\|$ and $\|\bar A_j x\|$ are incomparable, in some
directions each is larger than the other.  Second, there is no utility
to change the right singular vectors of $\hat A_j$ to align with those
of $A_j$.  The skew between the two can be arbitrarily large,
again see Figure \ref{fig:ellipse}, and without being able to adjust
these, this error can not be bounded.

One option would be to after every $\sqrt{m}$ rounds send a Frequent Directions
sketch $B_j$ of $A_j$ of size $O(1/\eps)$ rows from each site to the coordinator.
Then we use this $B_j$ as the new $\hat A_j$.  This has two problems.
First it only has $O(1/\eps)$ singular vectors that are well-defined,
so if there is increased squared norm in its null space, it is not
well-defined how to update it.  And second, still in between these
updates within a round, there is no way to maintain the error.

One can also try to shorten a round to update when $\hat F$ increases
by a $(1+\eps)$ factor, to bound the change within a round.  But this
causes $O((1/\eps) \log (\beta N))$ rounds, as in Section
\ref{sec:P2m}, and leads to $O((\sqrt{m}/\eps^2) \log (\beta N))$
total messages, which is as bad as the very conservative and deterministic algorithm \pim.  
Thus, for direction $x$ that is a right singular value as analyzed
above, we can get a Protocol 4 with $O((\sqrt{m}/\eps) \log (\beta
N))$ communication. But in the general case, how to, or if it is
possible at all to, get $O((\sqrt{m}/\eps) \log (\beta N))$
communication, as Protocol 4 does for weighted heavy hitters, in
arbitrary distributed matrix tracking is an intriguing open problem.

\Paragraph{Experiments with \piiiiw.}
In order to give a taste on why \piiiiw\ does not work, we compared it with other protocols. Figures \ref{fig:p3_pamap} and \ref{fig:p3_msd} show the the error this protocol incurs on PAMAP and MSD datasets.  Not only does it tend to accumulate error for smaller values of $\eps$, but for the PAMAP dataset and small $\eps$, the returned answer is almost all error. 
\begin{figure}
\begin{centering}
\subfigure[\s{err} vs. $\eps$.]{
\includegraphics[width=\figsize]{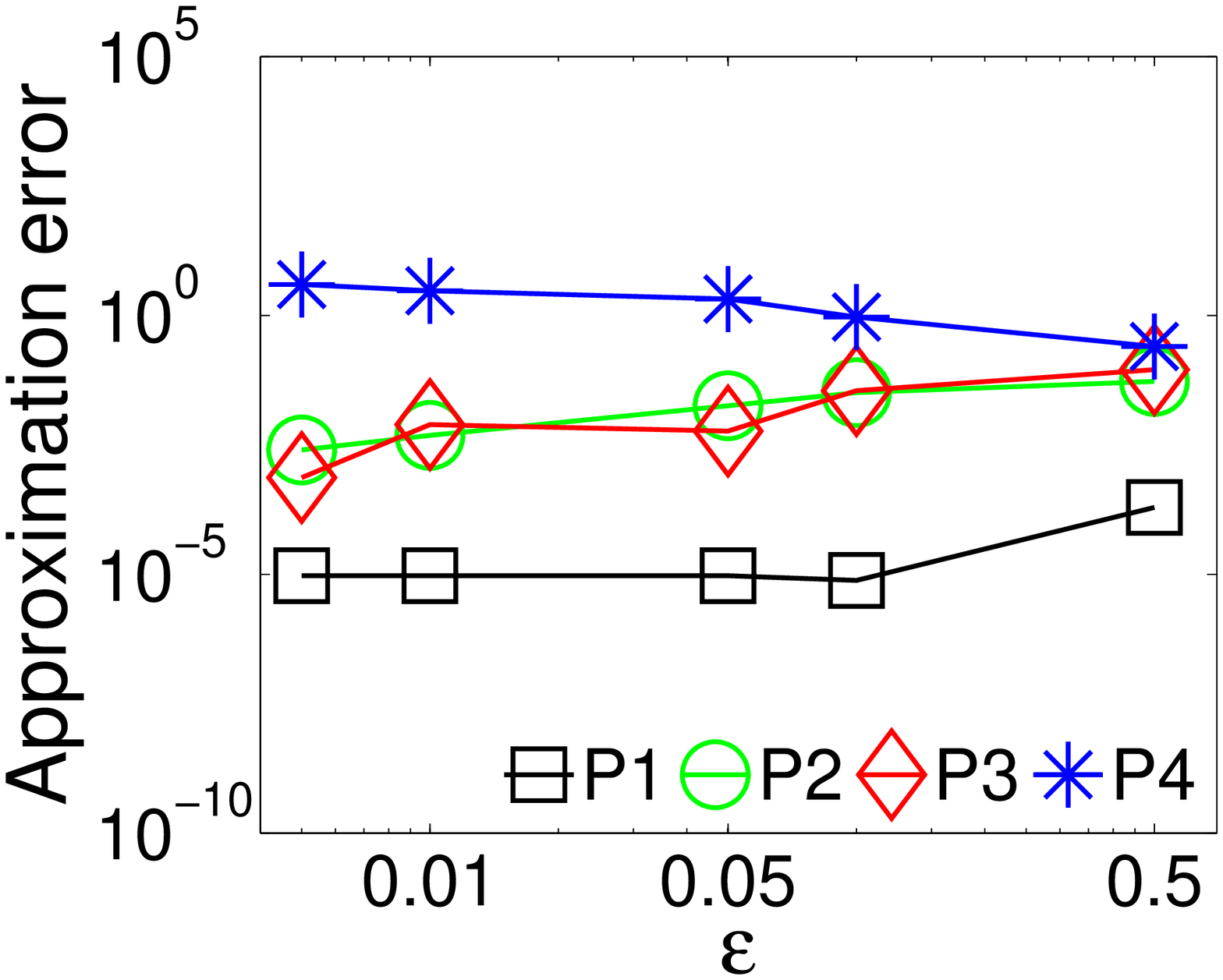}
\label{fig:pamap_err_eps_p3}
}
\subfigure[\s{err} vs. site]{
\includegraphics[width=\figsize]{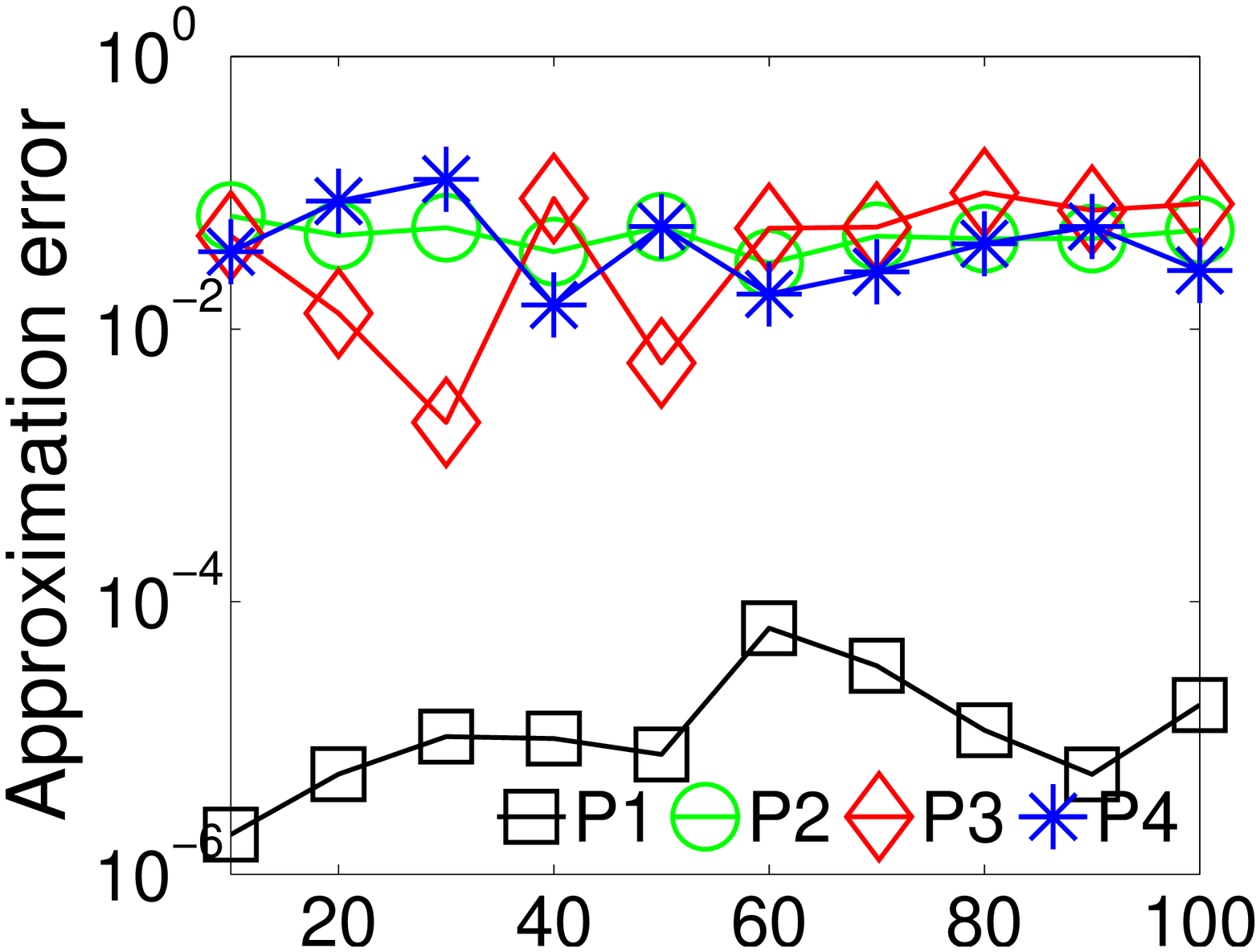}
\label{fig:pamap_err_sites_p3}
}
\caption{P4 vs. other protocols on PAMAP}\vspace{-3mm}
\label{fig:p3_pamap}
\end{centering}
\end{figure}

\begin{figure}
\begin{centering}
\subfigure[\s{err} vs. $\eps$.]{
\includegraphics[width=\figsize]{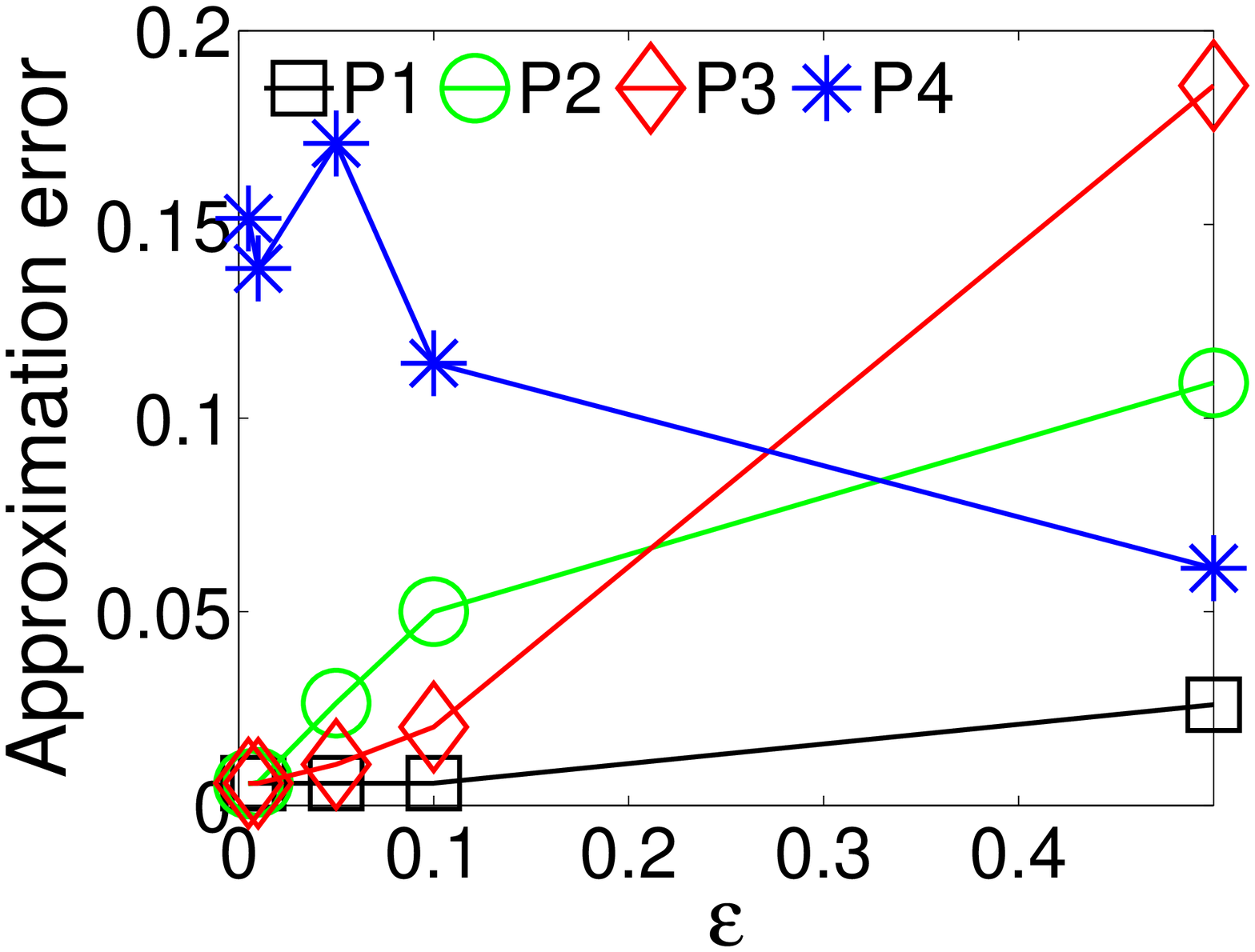}
\label{fig:yp_err_eps_p3}
}\vspace{2mm}
\subfigure[\s{err} vs. site]{
\includegraphics[width=\figsize]{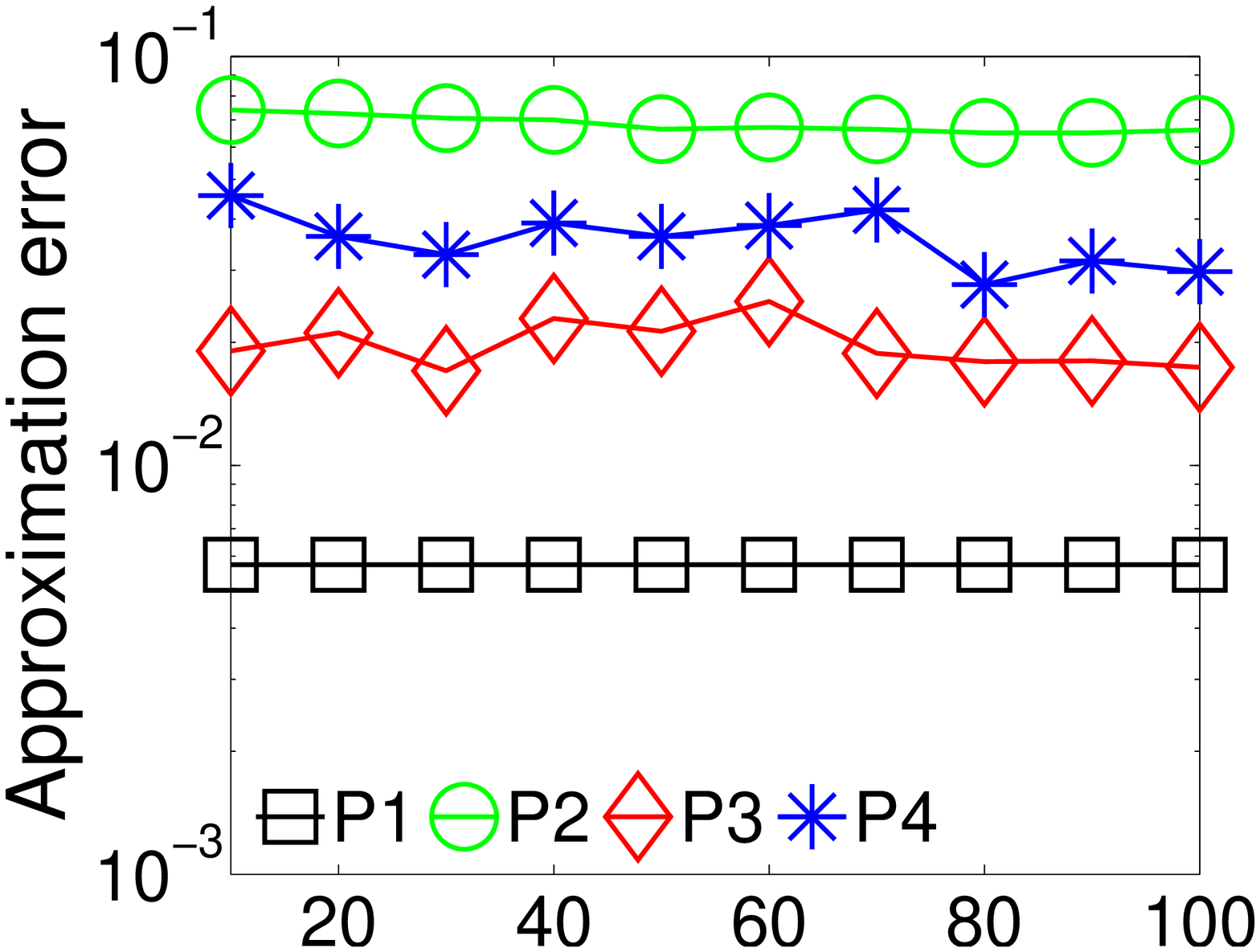}
\label{fig:yp_err_sites_p3}
}
\caption{P4 vs. other protocols on MSD}\vspace{-3mm}
\label{fig:p3_msd}
\end{centering}
\end{figure}

\end{appendix}

\end{document}